\documentclass{LMCS}

\usepackage{latexsym,hyperref}

\newcommand{\oomit}[1]{}

\newcommand{\dom}{\mathrm{dom}}
\newcommand{\tr}{\mathsf{t}}
\newcommand{\probab}{\mathsf{p}}
\newcommand{\Mono}{\mathit{Mono}}

\newcommand{\Cn}{\mathit{Cn}}

\newcommand{\PDC}{\mathit{PDC}}
\newcommand{\DC}{\mathit{DC}}
\newcommand{\PNL}{\mathit{PNL}}
\newcommand{\NL}{\mathit{NL}}
\newcommand{\ITL}{\mathit{ITL}}
\newcommand{\LTL}{\mathit{LTL}}
\newcommand{\CTL}{\mathit{CTL}}
\newcommand{\PITL}{\mathit{PITL}}
\newcommand{\infin}{\mathit{inf}}
\newcommand{\fin}{\mathit{fin}}

\newcommand{\dlceil}{\lceil\hspace{-0.05in}\lceil}
\newcommand{\drceil}{\rceil\hspace{-0.05in}\rceil}
\newcommand{\pred}[1]{\dlceil #1 \drceil}
\newcommand{\sem}[1]{[\![#1]\!]}
\newcommand{\synsem}[1]{(\!(#1)\!)}
\def \cat {{}^{\frown}\!}
\def\eqalign#1{\null\,\vcenter{\openup\jot\mathsurround=0 pt
  \ialign{\strut\hfil$\displaystyle{##}$&$\displaystyle{{}##}$\hfil
      \crcr#1\crcr}}\,}

\def\doi{3 (3:3) 2007}
\lmcsheading%
{\doi}
{1--43}
{}
{}
{May\phantom{.}~\phantom{0}5, 2005}
{Jul.~19, 2007}
{}

\begin{document}

\title[Probabilistic $\ITL$ and $\DC$ with Infinite Intervals: Complete Proof Systems]{Probabilistic Interval Temporal Logic and\\ Duration Calculus
with Infinite Intervals:\\
Complete Proof Systems }

\author[D.~P.~Guelev]{Dimitar~P.~Guelev}	
\address{Institute of Mathematics and Informatics, Bulgarian Academy of Sciences}	
\email{gelevdp@math.bas.bg}  


\begin{abstract}
\noindent 
The paper presents probabilistic extensions of interval temporal logic
($\ITL$) and duration calculus ($\DC$) with infinite intervals and
complete Hilbert-style proof systems for them. The completeness
results are a strong completeness theorem for the system of
probabilistic $\ITL$ with respect to an abstract semantics and a
relative completeness theorem for the system of probabilistic $\DC$
with respect to real-time semantics. The proposed systems subsume
probabilistic real-time $\DC$ as known from the literature. A
correspondence between the proposed systems and a system of
probabilistic interval temporal logic with finite intervals and
expanding modalities is established too.
\end{abstract}

\maketitle

\section*{Introduction}

The {\em duration calculus} ($\DC$) was introduced by Zhou, Hoare and
Ravn in \cite{ZHR91} as a logic to specify requirements on real-time
systems. $\DC$ is a classical predicate interval-based linear-time
logic with one normal binary modality known as {\em chop}. $\DC$ was
originally developed for real time by augmenting the real-time variant
of {\em interval temporal logic} ($\ITL$, \cite{Mos85,Mos86}) with
boolean expressions for {\em state} and real-valued terms to denote
state {\em durations}. $\DC$ has been used successfully in many case
studies such as \cite{ZZ94,DW94,SX98,DVH98,LH99}. We refer the reader
to \cite{HZ97} or the recent monograph \cite{ZH04} for a comprehensive
introduction to $\DC$.

Temporal logics such as linear temporal logic ($\LTL$), computation
tree logic ($\CTL$) and their timed versions are used mostly as
requirements languages for model-checkers such as SMV \cite{McMa} and
UPPAAL \cite{UPPAAL} which accept descriptions of systems in dedicated
input languages. The probabilistic variant of $\CTL$ \cite{ASB95} has
a similar role in the probabilistic model checker PRISM
\cite{KNP01,PRISM}. The systems in use are typically propositional,
which restricts the variety of properties that can be expressed. This
is only in part compensated for by the possibility to do fully
algorithmic verification. More complex properties and systems which,
e.g., involve unspecified numbers of concurrent processes or unbounded
amounts of data have to be viewed as parameterized families and
require the development of dedicated techniques. Alternatively,
model-checkers are used on instances of the systems with artificial
bounds on their size, which, however, quickly leads to the notorious
{\em state space explosion} problem. The use of the logics as {\em
reasoning tools} and not just as {\em notations} is also limited to
optimising simplifications such as abstractions. Unlike these systems
of logic, the expressive power of $\DC$ is geared towards the
possibility to capture the semantics of the systems to be verified and
therefore it is used as a system description language as
well. Examples include the $\DC$ semantics of the timed specification
language RAISE proposed in \cite{LH99} and the $\DC$ semantics of the
Verilog hardware specification language \cite{IEEE95} proposed in
\cite{SX98}. This shifts the interest from the satisfaction of $\DC$
formulas by given models towards validity in $\DC$.

The needs of applications have brought to life a number of extensions
and variants of $\DC$. These include state quantifiers and the least
fixed point operator \cite{Pan95}, alternative sets of interval
modalities \cite{Pan96,ZH96adequate,BRZ00,He99b}, enhancements of the
semantics to combine real and discrete time \cite{PD97,He99,Gue04} and
infinite intervals \cite{ZDL95,PWX98,SX98,WX04}. The extension of
$\DC$ by a probability operator replaces the linear model of time of
$\DC$ by a model based on sets of behaviours with probability on
them. Despite the absence of an explicit branching-time modality, the
probabilistic $\DC$ ($\PDC$) is essentially a branching-time predicate
interval-based temporal logic.

$\DC$ and, consequently, its extensions are not recursively
axiomatisable.  The worst case complexity of decision procedures for
validity is high even for very restricted subsets of $\DC$ such as the
so-called propositional $\DC$ \cite{ZHS93,Rab98}. No interesting
quantified decidable subsets of $\DC$ seem to be known (The state
quantifier in the $\lceil P\rceil$-subset of $\DC$ studied in
\cite{ZHS93} is expressible in that subset and does not increase its
ultimate expressive power.) The propositional abstract-time and
real-time $\ITL$s with {\em chop} are undecidable too. Undecidability
is typical of interval-based systems as shown in the early works
\cite{HS86} and \cite{Ven91a,Ven91thesis} where the {\em chop}
modality was studied as an example of an operator in many-dimensional
modal logic. A very simple subset of $\DC$ which exhibits its
incompleteness was identified in \cite{Gue04c}. This is compensated by
the convenience of achieving composionality in specification and
particularly the specification of sequential composition, which is
deemed to be difficult to handle in systems without the {\em chop}
modality \cite{Mull99}. Tool support for $\ITL$ and $\DC$ has been
developed on the basis of PVS \cite{PVS} by combining $\ITL$- and
$\DC$-specific proof and proof through translation into the
higher-order logic input language of PVS
\cite{SS94,Hu99thesis,Ras02}. There is also a model- and
validity-checker DCVALID \cite{DCVALID}, which accepts the discrete
time $\lceil P\rceil$-subset of $\DC$ ($QDDC$) and a combination of
$QDDC$ with $\CTL^*$ \cite{Pan01} and uses MONA \cite{Mona} as a
back-end tool. The expressive power of these subsets of $\DC$ is that
of weak monadic second order logic with one successor
($WS1S$). DCVALID has been successful in interesting case studies such
as that from \cite{Pan02}. However, the finite-state-based algorithms
of MONA impose on it the same ultimate limitations as in other
model-checking tools. That is why proof systems are a relatively
important instrument for verification by $\DC$ and its extensions.

$\DC$ was originally introduced for real time, whereas $\PDC$ was
first introduced in \cite{LRSZ92} for discrete time. A system of
real-time $\PDC$ was introduced later in \cite{DZ99} where some axioms
were proposed too. However, these axioms do not form a complete proof
system. Calculation with direct reference to the semantics was used to
reason about properties expressed in $\PDC$ in both works. More case
studies in $\PDC$ were given in \cite{Jos95} and recently in
\cite{ZH04}, which contains a chapter on discrete time $\PDC$. The
deductive power of the proof system for discrete time $\PDC$ used in
\cite{ZH04} has not been studied either.

A first attempt to develop a complete proof system for $\PDC$ was made
in \cite{Gue98probab}, where a system of probabilistic $\ITL$ was
proposed with the $\DC$-specific state expressions with finite
variability withdrawn. However, the semantics of that logic had some
non-standard elements for technical reasons, and the proof system was
a mixture of $\ITL$ and elements from Neighbourhood Logic ($\NL$,
\cite{ZH96adequate,RZ97,BRZ00}). Some of these problems were
eliminated in \cite{Tri99}. A more streamlined system of probabilistic
$\NL$ and a complete proof system with respect to its abstract-time
semantics was proposed later in \cite{Gue00probab}. The use of a
(commutative) linearly-ordered group as the model of time in that
system after Dutertre's work on abstract-time $\ITL$ \cite{Dut95}
allowed a finitary complete proof system to be obtained. However,
$\PNL$ still had some loose ends; the questions of the precise
correspondence between $\PNL$ and the original systems of $\PDC$ from
\cite{LRSZ92,DZ99} and of the deductive power of the proof system with
respect to real-time models remained open. Systems of
(non-probabilistic) branching time $\NL$ were developed in the recent
works \cite{BMS07} and \cite{BM05}. Some of these systems can be
viewed as the underlying branching time logics of $\PNL$. The works
\cite{BMS07} and \cite{BM05} present the propositional variants of
these branching time interval temporal logics and focus on decision
procedures for them.

In this paper we first propose another system of probabilistic
$\ITL$. Unlike that from \cite{Gue98probab}, this system is based on
infinite intervals. We propose a proof system for probabilistic $\ITL$
with infinite intervals which is complete with respect to the
abstract-time semantics based on that for $\ITL$ with infinite
intervals from \cite{WX04}. The use of infinite intervals removes the
need to admix $\NL$ modalities in proofs, which was done in
\cite{Gue98probab}. Then we develop a system of probabilistic $\DC$
($\PDC$) as an extension of the proposed probabilistic $\ITL$ and
demonstrate that adding the $\DC$ axioms and rules known from
\cite{HZ92} to our proof system for this probabilistic $\ITL$ leads to
a proof system for $\PDC$ with is complete with respect to real-time
models relative to validity at the real-time-based frame in
probabilistic $\ITL$ with infinite intervals. The incompleteness of
$\DC$ implies that relative completeness like that from \cite{HZ92}
for basic $\DC$ is the best we can have with a finitary proof
system. Finally, we describe satisfaction-preserving translations
between $NL$-based $\PDC$ and the system of $\PDC$ with infinite
intervals that we propose.

Our system of $\PDC$ has some slight enhancements in comparison with
the original probabilistic $\DC$ from \cite{LRSZ92,DZ99}. They both
improve its expressivity and facilitate the design of the proof
system. The first enhancement is a simplification. We remove the extra
reference time point needed to define the probability operator. The
role of this time point is naturally transferred to the flexible
constant $\ell$ which expresses interval lengths in $\DC$. This
extends the possibilities for meaningful nesting of occurrences of the
probability operator and allows the expression of probabilities of
properties which are probabilistic themselves. The second enhancement
is the use of infinite intervals. It is a consequence of our
developing of $\PDC$ as an extension of an infinite-interval-based
system of probabilistic $\ITL$. As mentioned above, this makes it
possible to avoid the use of an expanding modality such as those of
$\NL$, which was made in \cite{Gue00probab}. The combination of the
{\em chop} modality and infinite intervals has the expressive power of
expanding modalities with the advantage of keeping the {\em
introspectivity} of {\em chop}, which is a technically useful
property. We discuss the trade-offs between $\NL$ and $\ITL$ in
Section \ref{itlvsnl}. The last enhancement is the replacement of the
probabilistic timed automata which were used in \cite{DZ99} to define
sets of behaviours and the respective probability functions for $\PDC$
models by arbitrary systems of probability functions, which can be
constrained by additional axioms in $\PDC$ theories. One such
constraint that we study in detail is the requirement on all the
probability functions in a model to be consistent with a global
probability function which is defined on the space of all the
behaviours of the modelled system. Models which describe the behaviour
of automata like those involved in the definition of the original
system of real-time $\DC$ from \cite{DZ99} can be described by $\PDC$
theories in this more general setting too.

\subsubsection*{Structure of the paper}

After the necessary preliminaries on $\ITL$ with infinite intervals
and $\DC$ we introduce our system of probabilistic $\ITL$ with
infinite intervals and a proof system for it. We prove the
completeness of this proof system with respect to the abstract
semantics of probabilistic $\ITL$, which is the main result of the
paper. Then we propose axioms which constrain the system of
probability functions in models of $\PITL$ to be consistent with a
global probability function to the extent that this constraint can be
formulated in the setting of abstract probabilies. In the rest of the
paper we introduce a system of probabilistic $\DC$ as an extension of
the new system of probabilistic $\ITL$ by state expressions and
duration terms for them based on the real-time frame of probabilistic
$\ITL$. We show how this system of $\PDC$ subsumes the system proposed
in \cite{DZ99}. The main result about $\PDC$ is the completeness of
the well-known axioms of $\DC$ from \cite{HZ92} relative to validity
in real-time and -probability-based models for probabilistic
$\ITL$. Before concluding the paper we explain the correspondence
between $\PNL$ from \cite{Gue00probab} and the infinite-interval based
$\PITL$ proposed in this paper. We conclude by explaining some of the
limitations of the scope of its main results.

\section{Preliminaries}

In this section we give preliminaries on $\ITL$ and $\DC$ with
infinite intervals as known from \cite{ZDL95,PWX98,SX98,WX04} and the
probability operator of $\PDC$ as introduced in \cite{LRSZ92,DZ99}.

\subsection{Interval temporal logic with infinite intervals}

Here follows a brief formal introduction to $\ITL$ with infinite
intervals as presented in \cite{WX04}, which extends the finite
interval abstract-time system of $\ITL$ proposed and studied in
\cite{Dut95}.

\subsubsection{Language}
\label{itllanguage}

An $\ITL$ {\em vocabulary} consists of {\em constant symbols} $c, d,
\ldots$, {\em individual variables} $x, y, z, \ldots$, {\em function
symbols} $f, g, \ldots$ and {\em relation symbols} $R,
\ldots$. Constant, function and relation symbols can be either {\em
rigid} or {\em flexible}. Below it becomes clear that rigid symbols
have the same meaning at all times, whereas the meaning of flexible
symbols can depend on the reference time interval. The rigid constants
$0$ and $\infty$, addition $+$, equality $=$, the flexible constant
$\ell$, which always evaluates to the length of the reference
interval, and a countably infinite set of individual variables are
mandatory in every $\ITL$ vocabulary. We denote the {\em arity} of
function and relation symbols $s$ by $\# s$.

Given a vocabulary, the definition of an $\ITL$ language is
essentially that of its sets of {\em terms} $t$ and {\em formulas}
$\varphi$, which can be defined by the following BNFs:

\begin{tabular}{lll}
$t$ & $::=$ & $c\mid x\mid f(t,\ldots,t)$\\ $\varphi$ & $::=$ &
$\bot\mid
R(t,\ldots,t)\mid(\varphi\Rightarrow\varphi)\mid(\varphi;\varphi)\mid\exists
x\varphi$\\
\end{tabular}

\noindent
Many authors use the alternative notation $\varphi\cat\psi$ for
formulas $(\varphi;\psi)$ which are built with the {\em chop}
modality.

Terms and formulas with no occurrences of flexible symbols are called
{\em rigid}. Other terms and formulas are called {\em flexible}. The
set of the variables which have free occurrences in a formula
$\varphi$ is denoted by $FV(\varphi)$.

\subsubsection{Models and satisfaction}

A finite interval $\ITL$ frame consists of a linearly ordered set
$\langle T,\leq\rangle$ called the {\em time domain}, a monoid
$\langle D,0,+\rangle$ called the {\em duration domain} and a function
$m:{\bf I}(T)\rightarrow D$ called the {\em measure function}, where
\[{\bf I}(T)=\{[\tau_1,\tau_2]:\tau_1,\tau_2\in T,\tau_1\leq\tau_2\}\]
is the set of the closed and bounded intervals in $T$. The monoid
$\langle D,0,+\rangle$ is required to satisfy some additional
axioms. The full list of axioms is:

\begin{tabular}{ll}
$(D1)$ & $x+(y+z)=(x+y)+z$\\
$(D2)$ & $x+0=0+x=x$\\
$(D3)$ & $x+y=x+z\Rightarrow y=z,\ x+z=y+z\Rightarrow x=y$\\
$(D4)$ & $x+y=0\Rightarrow x=y=0$\\
$(D5)$ & $\exists z(x+z=y\vee y+z=x),\ \exists z(z+x=y\vee z+y=x)$\\
\end{tabular}

\noindent
The measure function $m$ is required to satisfy the axioms:

\begin{tabular}{ll}
$(M1)$ & $m([\tau_1,\tau_2])=m([\tau_1,\tau_2'])\Rightarrow \tau_2=\tau_2'$\\
$(M2)$ & $m([\tau_1,\tau])+m([\tau,\tau_2])=m([\tau_1,\tau_2])$\\
$(M3)$ & $m([\tau_1,\tau_2])=x+y\Rightarrow \exists \tau (m([\tau_1,\tau])=x)$\\
\end{tabular}

In the case of $\ITL$ with infinite intervals the time domain $\langle
T,\leq\rangle$ is supposed to have a distinguished greatest element
$\infty$ and $m$ is defined on the set $\tilde{\bf I}(T)={\bf
I}^{\fin}(T)\cup{\bf I}^\infin(T)$, where
\[{\bf I}^{\fin}(T)=\{[\tau_1,\tau_2]:\tau_1,\tau_2\in T,\tau_1\leq\tau_2<\infty\}\mbox{ and }
{\bf I}^{\infin}(T)=\{[\tau,\infty]:\tau\in T,\tau<\infty\}.\]
The duration domain is augmented with a greatest element $\infty$ too. The axiom $D3$ is weakened to

\begin{tabular}{ll}
$(D3')$ & $x+y=x+z\Rightarrow x=\infty\vee y=z,\ x+z=y+z\Rightarrow z=\infty\vee  x=y$\\
\end{tabular}

\noindent
and the following axioms about durations and the measure functions are added:

\begin{tabular}{ll}
$(D6)$ & $x+y=\infty\Leftrightarrow x=\infty\vee y=\infty$\\
$(M4)$ & $m([\tau_1,\tau_2])=\infty$ iff $\tau_2=\infty$\\
\end{tabular}

Given $\sigma_1,\sigma_2\in\tilde{\bf I}(T)$ such that
$\max\sigma_1=\min\sigma_2$, we denote $\sigma_1\cup\sigma_2$ by
$\sigma_1;\sigma_2$.

A function $I$ on an $\ITL$ vocabulary ${\bf L}$ is an {\em
interpretation of ${\bf L}$ into a frame\\ ${F=\langle\langle
T,\leq,\infty\rangle,\langle D,+,0,\infty\rangle,m\rangle}$} if it
satisfies the conditions:

$I(c), I(x)\in D$ for rigid constants $c$ and individual variables $x$; 

$I(f)\in (D^{\# f}\rightarrow D)$ for rigid function symbols $f$;

$I(R)\in (D^{\# R}\rightarrow \{0,1\})$ for rigid relation symbols $R$;

$I(c)\in (\tilde{\bf I}(T)\rightarrow D)$, $I(f)\in (\tilde{\bf I}(T)\times D^{\# f}\rightarrow D)$, $I(R)\in (\tilde{\bf I}(T)\times D^{\# R}\rightarrow \{0,1\})$ for flexible $c$, $f$ and $R$;

$I(0)=0$, $I(\infty)=\infty$, $I(+)=+$, $I(=)$ is $=$ and $I(\ell)=m$. 

An infinite-interval model for an $\ITL$ vocabulary ${\bf L}$ is a
pair of the form $\langle F,I\rangle$ such that $F$ is a frame and $I$
is an interpretation of ${\bf L}$ into $F$.

\begin{defi}
Given a model $\langle F,I\rangle$, the values $I_\sigma(t)$ of terms
$t$ at intervals $\sigma\in\tilde{\bf I}(T)$ is defined by the
clauses:

\begin{tabular}{lcl}
$I_\sigma(x)$ & $=$ & $I(x)$ for individual variables $x$\\
$I_\sigma(c)$ & $=$ & $I(c)$ for rigid constants $c$\\
$I_\sigma(f(t_1,\ldots,t_{\# f}))$ & $=$ & $I(f)(I_\sigma(t_1),\ldots,I_\sigma(t_{\# f}))$ for rigid function symbols $f$\\ 
$I_\sigma(c)$ & $=$ & $I(c)(\sigma)$ for flexible $c$\\
$I_\sigma(f(t_1,\ldots,t_{\# f}))$ & $=$ & $I(f)(\sigma,I_\sigma(t_1),\ldots,I_\sigma(t_{\# f}))$ for flexible $f$\\
\end{tabular}
\end{defi}

In particular, $I_\sigma(\ell)=m(\sigma)$, which means that the
function on $\tilde{\bf I}$ which is the meaning of the flexible
constant $\ell$ always evaluates to the length of the reference
interval $\sigma$.

\begin{defi}
Let $I$ be an interpretation of some $\ITL$ vocabulary ${\bf L}$ into
a frame $F$ whose duration domain is $\langle
D,+,0,\infty\rangle$. Let $x$ be an individual variable in ${\bf L}$
and $d\in D$. Then the interpretation $J$ of ${\bf L}$ into $F$ which
is defined by the equalities
\[J(x)=d\mbox{ and }J(s)=I(s)\mbox{ for }s\in{\bf L}\setminus\{x\}\]
is denoted by $I_x^d$ and is called a {\em $x$-variant of $I$}. We
abbreviate $(\ldots(I_{x_1}^{d_1})_{x_2}^{d_2}\ldots)_{x_n}^{d_n}$ by
$I_{x_1,\ldots,x_n}^{d_1,\ldots,d_n}$ and call it an
$x_1,\ldots,x_n$-variant of $I$. An $x_1,\ldots,x_n$-variant of $I$
for some finite list of variables $x_1,\ldots,x_n$ is called just {\em
variant}.
\end{defi}

\noindent
The modelling relation $\models$ on models based on some frame $F$,
intervals $\sigma$ and formulas in the vocabulary ${\bf L}$ is defined
by the clauses:

\begin{tabular}{lp{4in}}
$\langle F,I\rangle,\sigma\not\models\bot$\\
$\langle F,I\rangle,\sigma\models R(t_1,\ldots,t_n)$ & iff
$I(R)(I_\sigma(t_1),\ldots,I_\sigma(t_n))=1$ for rigid $R$\\
$\langle F,I\rangle,\sigma\models R(t_1,\ldots,t_n)$ & iff
$I(R)(\sigma,I_\sigma(t_1),\ldots,I_\sigma(t_n))=1$ for flexible $R$\\
$\langle F,I\rangle,\sigma\models(\varphi\Rightarrow\psi)$ & iff either
$\langle F,I\rangle,\sigma\not\models\varphi$ or $\langle F,I\rangle,\sigma\models\psi$\\
$\langle F,I\rangle,\sigma\models(\varphi;\psi)$ & iff

$\langle F,I\rangle,\sigma_1\models\varphi$ and $\langle F,I\rangle,\sigma_2\models\psi$\\
& for some $\sigma_1\in{\bf I}^\fin(T_F)$ and $\sigma_2\in\tilde{\bf I}(T_F)$ such that $\sigma_1;\sigma_2=\sigma$\\
$\langle F,I\rangle,\sigma\models\exists x\varphi$ & iff
$\langle F,I_x^d\rangle,\sigma\models\varphi$ for some $d\in D$\\

\end{tabular}

\subsubsection{Abbreviations and precedence of operators}

The binary relation symbol $\leq$ is defined in $\ITL$ by the equivalence
\begin{equation}\label{leqdef}
x\leq y\Leftrightarrow\exists z(x+z=y).
\end{equation}
The customary {\em infix} notation for $+$, $\leq$ and $=$ is used in
$\ITL$. $\top$, $\wedge$, $\Rightarrow$ and $\Leftrightarrow$,
$\forall$, $\not=$, $\geq$, $<$ and $>$ are used in the usual way. We
denote the {\em universal closure} $\forall x_1\ldots\forall
x_n\varphi$ of a formula $\varphi$ where
$\{x_1,\ldots,x_n\}=FV(\varphi)$ by $\forall\varphi$.

Since $(.;.)$ is associative, we omit parentheses in formulas with
consecutive occurrences of $(.;.)$. Here follow the infinite-interval
versions of some $\ITL$ abbreviations:

$\Diamond\varphi\rightleftharpoons (\top; \varphi;
\top)\vee(\top;\varphi)$ ,
$\Box\varphi\rightleftharpoons\neg\Diamond\neg\varphi$ .

\noindent
Note that $\Box$ and $\Diamond$ abbreviate different constructs in the
original discrete-time system of $\ITL$ of Moszkowski. Our usage
originates from the literature on $\DC$. The disjunctive member
$(\top;\varphi)$ in the definition of $\Diamond$ is relevant only at
infinite intervals. The formula $(\top; \varphi; \top)$ without it
restricts the subinterval which satisfies $\varphi$ to be finite.

We assume that $\Diamond$ and $\Box$ bind more tightly and $(.;.)$
binds less tightly than the boolean connectives.

\subsubsection{Proof system}
\label{proofsystemitl}

A complete proof system for abstract-time $\ITL$ with finite intervals
is given in \cite{Dut95}. The following axioms and rules have been
shown to form a complete proof system for $\ITL$ with infinite
intervals when added to a Hilbert-style proof system for classical
first-order predicate logic and the axioms $D1$, $D2$, $D3'$,
$D4$-$D6$ about durations in \cite{WX04}:

\qquad

\begin{tabular}{ll}
$(A1)$ & $(\varphi; \psi)\wedge\neg(\chi;\psi)\Rightarrow
(\varphi\wedge\neg\chi; \psi)$, $(\varphi;\psi)\wedge\neg(\varphi; \chi)\Rightarrow
(\varphi; \psi\wedge\neg\chi)$\\
$(A2)$ & $((\varphi;\psi);\chi)\Leftrightarrow(\varphi;
(\psi;\chi))$\\
$(R)$ & $(\varphi;\psi)\Rightarrow\varphi$, $(\psi;\varphi)\Rightarrow\varphi$ if $\varphi$ is rigid\\
$(B)$ & $(\exists x\varphi; \psi)\Rightarrow\exists
x(\varphi; \psi)$, $(\psi; \exists x\varphi)\Rightarrow\exists
x(\psi; \varphi)$ if $x\not\in FV(\psi)$\\
$(L1)$ & $(\ell=x; \varphi)\Rightarrow\neg(\ell=x;
\neg\varphi)$, $(\varphi; \ell=x\wedge x\not=\infty)\Rightarrow\neg(\neg\varphi;
\ell=x)$\\
$(L2)$ & $\ell=x+y\wedge x\not=\infty\Leftrightarrow(\ell=x;\ell=y)$\\
$(L3)$ & $\varphi\Rightarrow(\ell=0;\varphi)$, $\varphi\wedge\ell\not=\infty\Rightarrow(\varphi;\ell=0)$\\
$(S1)$ & $(\ell=x\wedge\varphi;\psi)\Rightarrow\neg(\ell=x\wedge\neg\varphi;\chi)$\\
$(P1)$ & $\neg(\ell=\infty;\varphi)$\\
$(P2)$ & $(\varphi;\ell=\infty)\Rightarrow\ell=\infty$\\
$(P3)$ & $(\varphi;\ell\not=\infty)\Rightarrow\ell\not=\infty$\\
\\
\raisebox{1mm}{$(N)$} & $\displaystyle\frac{\varphi}{\neg(\neg\varphi;
  \psi)}$ , $\displaystyle\frac{\varphi}{\neg(\psi;
  \neg\varphi)}$\\
\\
\raisebox{1mm}{$(\Mono)$} & $\displaystyle\frac{\varphi\Rightarrow\psi}{(\varphi; \chi)\Rightarrow(\psi;\chi)}$ , $\displaystyle\frac{\varphi\Rightarrow\psi}{(\chi; \varphi)\Rightarrow(\chi;\psi)}$
\end{tabular}

\qquad

\noindent
The presence of the modality $(.;.)$ and flexible symbols in $\ITL$ brings a restriction on the use of first order logic axioms which involve substitution such as 
\[(\exists_r)\ [t/x]\varphi\Rightarrow\exists x\varphi.\]
The application of this axiom is correct only if no variable in $t$ becomes bound due to the substitution, and either $t$ is rigid or $(.;.)$ does not occur in $\varphi$, because the value of a flexible term could be different at the different intervals which are involved in evaluating formulas with $(.;.)$.

The correctness of the proof system can be established by a direct check. Here follow some comments and informal reading of the axioms and the proof rules which can be helpful for their understanding too. $A1$ states that if chopping into a $\varphi$-subinterval and a $\psi$-subinterval is possible, but chopping into a $\chi$-subinterval and a $\psi$-subinterval is not, then any chopping into a $\varphi$- and a $\psi$-subinterval would lead to a $\varphi$-subinterval which additionally satisfies the negation of $\chi$. In the presence of the rules $\Mono$ and propositional tautologies one can choose between $A1$ and the axiom
\[(\alpha;\psi)\vee(\beta;\psi)\Leftrightarrow(\alpha\vee\beta;\psi),\]
which can be described as {\em distributivity of $(.;.)$ over $\vee$.} Axiom $B$ can be viewed as an parametric analogon of this distributivity axiom, with $\exists x$ to be read as parametric (possibly infinitary) disjunction. $A2$ is just the associativity of $(.;.)$. $R$ states that the satisfaction of rigid formulas does not depend on the reference interval. $L1$ and $S1$ express that if, upon dividing an interval, the duration of one of the subintervals is fixed, then the properties of both subintervals are completely determined. This is so because the subintervals themselves are uniquely determined. $L2$ is the additivity of length. $P2$ and $P3$ give separate treatment to some special cases of additivity that arise from the presence of infinitely long intervals. $L3$ states that intervals of length $0$ can be assumed at either end of any interval. $P3$ rules out the interval $[\infty;\infty]$. The rules $N$ state that valid formulas are valid in subintervals too. These rules are the standard form of the modal logic rule $\varphi/\Box\varphi$, yet about the {\em binary} modality $(.;.)$. The fact that weakening the condition on a subinterval in a $(.;.)$-formula can only facilitate the satisfiability of the whole $(.;.)$-formula is expressed by the rules $\Mono$.

\subsection{$\DC$ with infinite intervals}
\label{dcprelim}

The formal definition of $\DC$ with infinite intervals as an extension of the logic of the real-time-based frame of $\ITL$ with infinite intervals below is after \cite{ZDL95}. The main feature of $\DC$ relative to $\ITL$ are state expressions which are propositional formulas that denote piece-wise constant $\{0,1\}$-valued functions of time. Unlike purely-$\ITL$ flexible symbols, $\DC$ state expressions denote functions on time {\em points} and not intervals.

\subsubsection{Language}

{\em $\DC$ vocabularies} are $\ITL$ vocabularies extended by {\em state variables} $P, Q, \ldots$. State variables are used to build {\em state expressions} $S$ which have the syntax

\begin{tabular}{rcl}
$S$ & $::=$ & ${\bf 0}\mid P\mid S\Rightarrow S$\\
\end{tabular}

\noindent
and in turn appear as the argument of {\em duration terms} $\int S$ which are the $\DC$-specific construct in the syntax of terms $t$:

\begin{tabular}{lll}

$t$ & $::=$ & $c\mid x\mid v\mid \int S\mid f(t,\ldots,t)$\\
\end{tabular}

\noindent
Duration terms are regarded as flexible. The syntax of formulas is as in $\ITL$.

Flexible constants and $0$-ary flexible predicate letters in $\DC$ are also known as {\em temporal variables} and {\em temporal propositional letters}, respectively.

\subsubsection{Semantics}
\label{semantics}

We are only interested in real-time $\DC$ which is based on the $\ITL$ frame \[F_{\bf R}=\langle\langle\overline{\bf R},\leq,\infty\rangle,\langle\overline{\bf R}_+,+,0,\infty\rangle,\lambda\sigma.\max\sigma-\min\sigma\rangle\]
where $\overline{\bf R}={\bf R}\cup\{\infty\}$ and $\overline{\bf R}_+=\{x\in \overline{\bf R}:x\geq 0\}$.

{\em $\DC$ interpretations} extend $\ITL$ interpretations to provide values for state variables, which are functions of type ${\bf R}\rightarrow\{0,1\}$ that satisfy the following {\em finite variability} requirement:

\begin{quote}
For every pair $\tau_1,\tau_2\in{\bf R}$ such that $\tau_1<\tau_2$, and every state variable $P$ there exist an $n<\omega$ and $\tau_1',\ldots,\tau_n'\in{\bf R}$ such that $\tau_1=\tau_1'<\ldots<\tau_n'=\tau_2$ and $I(P)$ is constant on the semi-open intervals $[\tau_i',\tau_{i+1}')$, $i=1,\ldots,n-1$.
\end{quote}

\noindent
Given an interpretation $I$, the values $I_\tau(S)$ of state expressions $S$ at time $\tau\in{\bf R}$ are defined by the equalities

\begin{tabular}{lcl}
$I_\tau({\bf 0})$ & $=$ & $0$\\
$I_\tau(P)$ & $=$ & $I(P)(\tau)$ for state variables $P$\\
$I_\tau(S_1 \Rightarrow S_2)$ & $=$ & $\max(1-I_\tau(S_1),I_\tau(S_2))$\\
\end{tabular}

\noindent
The value $I_\sigma(\int S)$ of duration term $\int S$ at interval $\sigma\in\tilde{\bf I}(\overline{\bf R})$ is defined by the equality
\[\textstyle I_\sigma(\int S)=\int\limits_{\min\sigma}^{\max\sigma}I_\tau(S)d\tau\]
Note that $\textstyle I_\sigma(\int S)$ can be $\infty$ for $\sigma\in{\bf I}^\infin(\overline{\bf R})$. The values of other kinds of terms and $\models$ are defined as in $\ITL$.

\subsubsection{Abbreviations}

The boolean connectives $\neg$, $\vee$, $\wedge$ and $\Leftrightarrow$ are used in state expressions as abbreviations in the usual way. The following abbreviations are specific to $\DC$:

\begin{tabular}{ll}
${\bf 1}\rightleftharpoons\neg{\bf 0}$\\
$\pred{S}\rightleftharpoons\int S=\ell\wedge\ell\not=0$\\
\end{tabular}

\noindent
Sometimes $\ell$ is introduced as an abbreviation for $\int{\bf 1}$.

\subsubsection{Proof system}
\label{dcinfax}

The axioms and rules below were proposed in \cite{HZ92} for $\DC$ with {\em finite} intervals.

\begin{tabular}{ll}
$(DC1)$ &$\int{\bf 0}=0$\\
$(DC2)$ &$\int{\bf 1}=\ell$\\
$(DC3)$ &$\int S\geq 0$\\
$(DC4)$ &$\int S_1+\int S_2=\int(S_1\vee S_2)+\int(S_1\wedge S_2)$\\
$(DC5)$ &$(\int S=x; \int S=y)\Rightarrow \int S=x+y$\\
$(DC6)$ &$\int S_1=\int S_2$ if $S_1$ and $S_2$ are propositionally equivalent\\
\raisebox{1mm}{$(IR1)$} &$\displaystyle\frac{[\ell=0/A]\varphi
\ \varphi\Rightarrow[A\vee(A;\pred{S}\vee \pred{\neg S})/A]\varphi}{[\top/A]\varphi}$\\
\raisebox{1mm}{$(IR2)$} &$\displaystyle\frac{[\ell=0/A]\varphi
\ \varphi\Rightarrow[A\vee(\pred{S}\vee \pred{\neg S};
A)/A]\varphi}{[\top/A]\varphi}$\\
\end{tabular}

These axioms and rules have been shown to be complete with respect to the finite-interval variant $\langle\langle{\bf R},\leq\rangle,\langle{\bf R}_+,+,0\rangle,\lambda\sigma.\max\sigma-\min\sigma\rangle$ of $F_{\bf R}$ relative to validity in the class of the $\ITL$ models which are based on the finite-interval variant of $F_{\bf R}$ in \cite{HZ92}.

The correctness of $IR1$ and $IR2$ is based on the finite variability of state. Since every finite interval can be partitioned into finitely many subintervals in which the state expression $S$ is constant, proving the validity of a property $\varphi$ about zero-length intervals and proving that the validity of $\varphi$ at intervals with $n$ alternations of the value of $S$ implies the validity of the same property about intervals with $n+1$ such alternations is sufficient to conclude that $\varphi$ holds about intervals with any finite number of alternations of the value of $S$. This, by the assumption of finite variability, means that $\varphi$ is valid about all intervals. The completeness proof from \cite{HZ92} involves two theorems which can be derived using the rules $IR1$ and $IR2$, instead of the rules themselves. The second of these theorems does not hold for infinite intervals and therefore we modify it appropriately:

\begin{tabular}{ll}
$(T1)$ & $\ell=0\vee(\pred{S};\top)\vee(\pred{\neg S};\top)$\\
$(T2)$ & $\ell=0\vee\ell=\infty\vee(\top;\pred{S})\vee(\top;\pred{\neg S})$\\
\end{tabular}

\noindent
The use of $T1$ and $T2$ instead of $IR1$ and $IR2$ brings technical convenience to the representation of $\DC$ as a theory in $\ITL$ with $DC1$-$DC6$, $T1$ and $T2$ as its axioms in the proof of relative completeness. 

We take $DC1$-$DC6$, $T1$ and the infinite-interval version of $T2$ as axioms to form a relatively complete proof system for $\DC$ with infinite intervals and disregard the rules $IR1$ and $IR2$ in the rest of the paper. The proof of the relative completeness of this system follows closely the pattern of the original proof from \cite{HZ92}. It appears as part of the proof of the relative completeness of our infinite-interval-based system of probabilistic $\DC$ in Section \ref{relativecompleteness}.

\subsection{Probabilistic $\DC$ for real time}

Probabilistic $\DC$ was first introduced for discrete time in \cite{LRSZ92}. There is a chapter on discrete time probabilistic $\DC$ in \cite{ZH04} too. Here follows the formal definition of real-time probilistic $\DC$ as introduced in \cite{DZ99}. 

\subsubsection{Real-time probabilistic automata}

The semantics of the real-time probabilistic $\DC$ as originally proposed in \cite{DZ99} is based on a class of real-time probabilistic automata.

\begin{defi}
\label{pautomata}
A {\em finite probabilistic timed automaton} is a system of the form
\begin{equation}\label{fpta}
{\bf
A}=\langle S,A,s_0,\langle q_a,a\in A\rangle,\langle
p_a:a\in A\rangle\rangle
\end{equation}
where:

$S$ is a finite set of {\em states};

$A\subset\{\langle s,s'\rangle:s,s'\in S,s\not=s'\}$ is a set of {\em
transitions};

$s_0\in S$ is called the {\em initial state};

$q_a\in [0,1]$ is the {\em choice probability} for transition $a\in A$;

$p_a\in ({\bf R}_+\rightarrow{\bf R}_+)$ is the {\em duration
probability density} of transition $a$.

Given the automaton ${\bf A}$, $A_s$ denotes $\{s'\in
S:\langle s,s'\rangle\in A\}$. If $a\in A$ and $a=\langle s,s'\rangle$, then $s$ and $s'$ are denoted by $a^-$ and $a^+$, respectively. Choice probabilities $q_a$ are required to satisfy $\sum\limits_{a\in A_s}q_a=1$ for $A_s\not=\emptyset$. Probability densities $p_a$ are required to satisfy $\int\limits_0^\infty p_a(\tau)d\tau=1$.
\end{defi}

An automaton ${\bf A}$ of the form (\ref{fpta}) works by going through a finite or infinite sequence of states $s_0$, $s_1$, \ldots , $s_n$, \ldots  such that
$\langle s_i,s_{i+1}\rangle\in A$ for all $i$. Each transition has a duration 
$d_i$, which is the time that elapses before $s_i$ changes to $s_{i+1}$. Thus individual {\em behaviours} of ${\bf A}$ can be represented as sequences of the form
\begin{equation}
\label{beh}
\langle a_0,d_0\rangle,\ldots,\langle a_n,d_n\rangle,\ldots
\end{equation}
where $a_i\in A$, $d_i\in {\bf R}_+$, $a^-_0=s_0$ and $a^+_i=a^-_{i+1}$ for all $i$. Having arrived at state $s$, ${\bf A}$ chooses transition $a\in A_s$ with probability $q_a$. The probability for the duration of $a$ to be in $[\tau_1,\tau_2]$ is $\int\limits_{\tau_1}^{\tau_2} p_a(\tau)d\tau$. 

Automata of the above type are closely related to the probabilistic real-time processes known from \cite{ACD91a,ACD91b}. 

\subsubsection{$\DC$ models for real-time probabilistic automata behaviours}

Probabilistic $\DC$ was introduced in \cite{DZ99} for vocabularies built to describe the behaviours of given real-time probabilistic automata. The $\DC$ vocabulary ${\bf L}_{\bf A}$ for (\ref{fpta}) has the states $s\in S$ as its state variables. The only other non-logical symbols are the mandatory ones. A $\DC$ interpretation of ${\bf L}_{\bf A}$ describes the behaviour (\ref{beh}) of ${\bf A}$ if for all $i<\omega$ $\tau\in\left[\sum\limits_{j<i}d_j,\sum\limits_{j\leq i}d_j\right)$ implies that $I_\tau(s_k)=1$ just for $k=i$.

\subsubsection{Satisfaction probability of $\DC$ formulas and probabilistic $\DC$ for real time}

Given a real-time probabilistic automaton (\ref{fpta}), the set ${\bf W}_{\bf A}$ of all the interpretations of ${\bf L}_{\bf A}$ which describe possible behaviours of ${\bf A}$ can be endowed with a probability function $\mu_{\bf A}$. Given $A\subseteq {\bf W}_{\bf A}$, $\mu_{\bf A}(A)$ can be defined as the probability for ${\bf A}$ to have a behaviour described by an interpretation in $A$. The sets $A$ in the domain of $\mu_{\bf A}$ should be chosen from some appropriate boolean algebra of subsets of $2^{{\bf W}_{\bf A}}$. Details on the definition of $\mu_{\bf A}$, including explicit formulas for $\mu_{\bf A}$ in terms of $p_a$ and $q_a$, can be found in \cite{DZ99}.

Given $\tau\in{\bf R}_+$ and a $\DC$ formula $\varphi$ in the vocabulary ${\bf L}_{\bf A}$, the value of  the {\em $\PDC$ term $\mu_{\bf A}(\varphi)(\tau)$} is defined as
\[\mu_{\bf A}(\{I\in {\bf W}_{\bf A}:I,[0,\tau]\models\varphi\}).\]
Probabilistic $\DC$ for real time was introduced in \cite{DZ99} by enhancing $\DC$ with terms of the form $\mu(\varphi)(t)$ where $\varphi$ is a $\DC$ formula in ${\bf L}_{\bf A}$ for some automaton ${\bf A}$ and $t$ is a term. The values of such terms were defined by the equality
\[I_\sigma(\mu(\varphi)(t))=\mu_{\bf A}(\varphi)(I_\sigma(t)).\]
Note that $I_\sigma(\mu(\varphi)(t))$ depends on $\sigma$ only through the value of $t$. This means that $\mu(\varphi)(t)$ is rigid iff $t$ is.

\section{Probabilistic $\ITL$ with infinite intervals}

In this section we extend abstract-time $\ITL$ with infinite intervals by a probability operator which generalises the operator $\mu(.)(.)$ of $\PDC$ from \cite{LRSZ92,DZ99}. The new probability operator is more expressive and syntactically simpler than $\mu(.)(.)$. Instead of the binary $\mu(\varphi)(t)$ we use a unary $p(\varphi)$ which takes the formula argument $\varphi$ of $\mu$. The semantics of $p(\varphi)$ given below makes it clear that the term argument $t$ which determines the length of the interval at which $\varphi$ is to be evaluated need not be written separately because $\mu(\varphi)(t)$ can be expressed as $p((\varphi\wedge\ell=t;\top))$. To accomodate the arithmetics of probabilities, abstract-time frames for the new system of probabilistic $\ITL$ include a similarly abstract probability domain. We use the acronym $\PITL$ for the new system. $\PITL$ and its proof system is the main topic of this paper. As it becomes clear below, $\PITL$ can be extended to $\PDC$ in a straightforward way.

\subsection{Language}

$\PITL$ vocabularies are {\em two-sorted}, with durations and probabilities being the two sorts. For this reason, instead of just arities, the non-logical symbols have {\em types} which determine the sorts of each argument in the cases of function and relation symbols, and the sort of terms built using the symbol for constants, variables and function symbols. A term or atomic formula $s(t_1,\ldots,t_{\# s})$ is well formed only if the sorts of the argument terms $t_1$, \ldots , $t_{\# s}$ match the type of $s$. 

Along with the mandatory non-logical symbols $0$, $\infty$, $+$ and $\ell$ of the duration sort, $\PITL$ vocabularies are required to include the rigid constants $0$ and $1$ and addition $+$ of the probability sort. Equality $=$ is included for each sort too. We use the same characters to denote these otherwise distinct symbols as long as this causes no confusion. We assume countably infinite sets of individual variables of either sort and no more than countably-infinite sets of other symbols in $\PITL$ vocabularies.

The syntax of $\PITL$ terms extends that from $\ITL$ by terms of the form $p(\varphi)$ where $\varphi$ is a formula. These terms are of the probability sort and we call them {\em probability terms}. $FV(p(\varphi))=FV(\varphi)$ and $p(\varphi)$ is rigid iff $\varphi$ is rigid. 

The syntax of formulas is as in $\ITL$. 

\subsection{Models and satisfaction}
\label{pitlsemantics}

The main part of a $\PITL$ model is a collection of interpretations of
the given vocabulary into a given two-sorted frame for $\ITL$ with
infinite intervals. These interpretations are meant to describe the
possible behaviours of a modelled system. Unlike the original $\PDC$
models, which assume a global probability function that is derived
from the laws of probabilistic behaviour of appropriate automata, we
assume a probability distribution to model the probabilistic branching
of every behaviour at every time point. Restrictions on the system of
probability distributions which, e.g., force them to model the choice
and duration probabilities of an appropriate automaton can be imposed
by additional axioms such as those from Section
\ref{modellingautomata}.

\begin{defi}
A $\PITL$ frame is a tuple of the form \[F=\langle\langle
T,\leq,\infty\rangle,\langle D,+,0,\infty\rangle,\langle
U,+,0,1\rangle,m\rangle\ ,\] where $\langle T,\leq,\infty\rangle$,
$\langle D,+,0,\infty\rangle$ and $m$ are as in frames for $\ITL$ with
infinite intervals and $\langle U,+,0,1\rangle$ is a commutative
monoid with the additional constant $1$, which is called the {\em
probability domain}. $\langle U,+,0,1\rangle$ is supposed to satisfy
some additional axioms. Here follows the full list:

\begin{tabular}{ll}
$(U1)$ & $x+(y+z)=(x+y)+z$\\
$(U2)$ & $x+y=y+x$\\
$(U3)$ & $x+0=x$\\
$(U4)$ & $x+y=x+z\Rightarrow y=z$\\
$(U5)$ & $x+y=0\Rightarrow x=y=0$\\
$(U6)$ & $\exists z(x+z=y\vee y+z=x)$\\
$(U7)$ & $0\not=1$\\

\end{tabular}
\end{defi}

We use the same symbols for $+$ and $0$ in both duration domains and probability domains, despite that they are different entities, as long as this causes no confusion. Probability domains are assumed to be ordered by the relation $\leq$ which is defined by (\ref{leqdef}) like in the case of durations.

For the rest of the section ${\bf L}$ denotes some $\PITL$ vocabulary and $F$ is some $\PITL$ frame with its components named as above.

\begin{defi}
A {\em $\PITL$ interpretation of ${\bf L}$ into $F$} is a function $I$ on ${\bf L}$ which satisfies the conditions:

$I(c), I(x)\in A$ for rigid constants $c$ and individual variables $x$ where $A$ is either $D$ or $U$, depending on the sort of the symbol; 

$I(f)\in (A_1\times\ldots\times A_{\# f}\rightarrow A_{{\# f}+1})$ for rigid function symbols $f$ where $A_1,\ldots,A_{{\# f}+1}$ are either $D$ or $U$ each, depending on the sort of the respective argument of $f$ and the sort of the value of $f$.

$I(R)\in (A_1\times\ldots\times A_{\# R}\rightarrow \{0,1\})$ for rigid relation symbols $R$ where $A_1,\ldots,A_{\# R}$ are chosen as for function symbols;

$I(c)\in (\tilde{\bf I}(T)\rightarrow A)$, $I(f)\in (\tilde{\bf I}(T)\times A_1\times\ldots\times A_{\# f}\rightarrow A_{{\# f}+1})$ and \\ $I(R)\in (\tilde{\bf I}(T)\times A_1\times\ldots\times A_{\# R}\rightarrow \{0,1\})$ for flexible $c$, $f$ and $R$ where the $A$s are chosen as for rigid symbols;

$I(0)=0$, $I(+)=+$ and $I(=)$ is $=$ for $0$, $+$ and $=$ of either sort and its corresponding domain in $F$. $I(1)$ is the constant $1$ from $U$. $I(\infty)=\infty$ and $I(\ell)=m$ like with $\ITL$ interpretations.
\end{defi}

Consider a non-empty set ${\bf W}$, a function $I$ on ${\bf W}$ into the set of the $\PITL$ interpretations of the fixed vocabulary ${\bf L}$ into the fixed frame $F$ and a function $P$ of type ${\bf W}\times T\times 2^{\bf W}\rightarrow U$. Let $I^w$ and $P^w$ abbreviate $I(w)$ and $\lambda\tau,X.P(w,\tau,X)$, respectively, for all $w\in{\bf W}$. $I^w$ and $P^w$, $w\in{\bf W}$, are intended to represent the set of behaviours and the associated probability distributions for every $\tau\in T$ in the $F$-based $\PITL$ models for ${\bf L}$ to be defined below.

\begin{defi}
Let $\tau\in T$. We define the equivalence relation $\equiv_\tau$ on ${\bf W}$ for all $\tau\in T$ by putting $w\equiv_\tau v$ iff

$I^w(s)=I^v(s)$ for all rigid symbols $s\in{\bf L}$, except possibly the individual variables;

$I^w(s)(\sigma,d_1,\ldots,d_{\# s})=I^v(s)(\sigma,d_1,\ldots,d_{\# s})$ for all flexible $s\in{\bf L}$, all $d_1,\ldots,d_{\# s}$ from the appropriate domains and all $\sigma\in\tilde{\bf I}(T)$ such that $\max\sigma\leq\tau$;

$P^w(\tau',X)=P^v(\tau',X)$ for all $X\subseteq{\bf W}$ and all $\tau'\leq\tau$.

\noindent
Given $w\in{\bf W}$ and $\tau\in T$, we denote the set
\[\{v\in{\bf W}:v\equiv_\tau w\}\]
by ${\bf W}_{w,\tau}$. 
\end{defi}

Members of ${\bf W}$ which are $\tau$-equivalent stand for the same behaviour up to time $\tau$. If $\tau_1>\tau_2$, then $\equiv_{\tau_1}\subset\equiv_{\tau_2}$ and $w\equiv_\infty v$ holds iff $P^w=P^v$ and $I^w$ and $I^v$ agree on all symbols, except possibly some individual variables. ${\bf W}_{w,\tau}$ is the set of those $v\in{\bf W}$ which represent the probabilistic branching of $w$ from time $\tau$ onwards.

\begin{defi}
\label{pitlmodels}
A {\em general $\PDC$ model} for ${\bf L}$ is a tuple of the form $\langle F,{\bf W},I,P\rangle$ where $F$, ${\bf W}$, $I$ and $P$ are as above and satisfy the following requirements for every $w\in{\bf W}$: 

{\em ${\bf W}$ is closed under variants of interpretations.} If $w\in{\bf W}$, $x$ is an individual variable from ${\bf L}$ and $a$ is in the domain from $F$ which corresponds to the sort of $x$, then there is a $v\in{\bf W}$ such that $P^v=P^w$ and $I^v=(I^w)_x^a$.

{\em $P^w$ represents probability measures.} The function $\lambda X.P^w(\tau,X)$ for every $w\in W$ and $\tau\in T$ is a finitely additive probability measure on the boolean algebra
\begin{equation}\label{powerset}
\langle 2^{\bf W},\cap,\cup,\emptyset,{\bf W}\rangle.
\end{equation}
and satisfies the equality
\[P^w(\tau,X)=P^w(\tau,X\cap {\bf W}_{w,\tau})\mbox{ for all }X\subseteq{\bf W},\]
which means that $\lambda X.P^w(\tau,X)$ is required to be {\em concentrated} on the set ${\bf W}_{w,\tau}$.
\end{defi}

Informally, a general $\PITL$ model is based on a set ${\bf W}$ of descriptions of infinite behaviours made by means of the $\ITL$ interpretations $I^w$ which are associated with each $w\in{\bf W}$. All the interpretations $I^w$ are into the same frame $F$ and are supposed to treat rigid symbols identically to express that, e. g., arithmetics is the same in all behaviours. It is assumed that, given a finite initial part of a behaviour $w$ until time $\tau$, the modelled system can proceed according to a description within the set ${\bf W}_{w,\tau}$ of the behaviours which are the same as $w$ up to time $\tau$. The probability for the system to choose a behaviour in $X\subseteq {\bf W}_{w,\tau}$ is $P^w(\tau,X)$. 

Next we define term values $w_\sigma(t)$ and the satisfaction of formulas in $\PITL$ models. The definitions of term values, the modelling relation $\models$ and its associated notation $\sem{.}$ for terms, formulas, models and time intervals in $\PITL$ are given by the following clauses, where the components of the model $M$ are named as above: 

\qquad

\noindent{\em Term values}

\qquad

\begin{tabular}{llll}
$w_\sigma(x)$ & $=$ & $I^w(x)$ & for variables $x$\\

$w_\sigma(c)$ & $=$ & $I^w(c)$ & for rigid $c$\\

$w_\sigma(f(t_1,\ldots,t_{\# f}))$ & $=$ & $I^w(f)(w_\sigma(t_1),\ldots,w_\sigma(t_{\#f}))$ & for rigid $f$\\
$w_\sigma(c)$ & $=$ & $I^w(c)(\sigma)$ & for flexible $c$\\

$w_\sigma(f(t_1,\ldots,t_{\# f}))$ & $=$ & $I^w(f)(\sigma,w_\sigma(t_1),\ldots,w_\sigma(t_{\#f}))$ & for flexible $f$\\
$w_\sigma(p(\psi))$ & $=$ & $P^w(\max\sigma,\sem{\psi}_{M,w,\sigma})$\\
\end{tabular}

\qquad

\noindent
Here $\sem{\psi}_{M,w,\sigma}$ stands for
\begin{equation}
\label{semvarphi}
\{v\in{\bf W}_{w,\max\sigma}:(\forall v'\in W)(P^{v'}=P^v\wedge I^{v'}=(I^v)_{x_1\ ,\ \ldots\ ,\ x_n}^{I^w(x_1),\ldots,I^w(x_n)}\rightarrow M,v',[\min\sigma,\infty]\models\psi
)\},
\end{equation}
where $x_1,\ldots,x_n$ are the free variables of $\psi$. This means that $\sem{\psi}_{M,w,\sigma}$ consists of the behaviours $v$ which are $\max\sigma$-equivalent to $w$ and satisfy $\psi$ at the infinite interval starting at $\min\sigma$.

\qquad

\noindent{\em Satisfaction of formulas}

\qquad

\noindent

\begin{tabular}{lp{3.8in}}
$M,w,\sigma\not\models\bot$\\
$M,w,\sigma\models R(t_1,\ldots,t_{\# R})$ & iff $I^w(R)(w_\sigma(t_1),\ldots,w_\sigma(t_{\#R}))=1$ for rigid $R$\\
$M,w,\sigma\models R(t_1,\ldots,t_{\# R})$ & iff $I^w(R)(\sigma,w_\sigma(t_1),\ldots,w_\sigma(t_{\#R}))=1$ for flexible $R$\\
$M,w,\sigma\models(\varphi\Rightarrow\psi)$ & iff either
$M,w,\sigma\not\models\varphi$ or $M,w,\sigma\models\psi$\\
$M,w,\sigma\models(\varphi;\psi)$ & iff
$M,w,\sigma_1\models\varphi$ and $M,w,\sigma_2\models\psi$\\
& for some $\sigma_1\in{\bf I}^\fin(T_F)$ and $\sigma_2\in\tilde{\bf I}(T_F)$ such that $\sigma_1;\sigma_2=\sigma$\\
$M,w,\sigma\models\exists x\varphi$ & iff
$M,v,\sigma\models\varphi$ for some $v\in{\bf W}$ and some $a$ from the domain of the sort of $x$ such that $P^v=P^w$ and $I^v=(I^w)^a_x$\\
\end{tabular}

\qquad

\noindent
Obviously $M,w,\sigma\models\psi$ iff $\langle F,I^w\rangle,[\min\sigma,\infty]\models_{\ITL}\psi$ as in non-probabilistic $\ITL$ for $\psi$ with no occurrence of probability terms. 

The probability functions $\lambda X.P^w(\tau,X)$ for $w\in{\bf W}$ and $\tau\in T$ in general $\PITL$ models $M=\langle F,{\bf W},I,P\rangle$ are needed just as much as they provide values for probability terms. That is why these functions need not be defined on the entire algebra (\ref{powerset}). Indeed, it is sufficient for $\lambda X.P^w(\tau,X)$ to be defined on the (generally smaller) algebra
\[\langle\{\sem{\psi}_{M,w,\sigma}:\psi\in{\bf L},\sigma\in\tilde{\bf I}(T),\max\sigma=\tau\},\cap,\cup,\emptyset,{\bf W}_{w,\tau}\rangle,\]
which we denote by ${\bf B}_{M,w,\tau}$. This observation justifies the broadening of the definition of general $\PITL$ models as follows.

\qquad 

\noindent
{\bf Amendment to Definition \ref{pitlmodels}} {\em Structures of the form $M=\langle F,{\bf W},P,I\rangle$ from Definition \ref{pitlmodels}, but with their probability functions $\lambda X.P^w(\tau,X)$ defined just on the respective algebras ${\bf B}_{M,w,\tau}$, are general $\PITL$ models too.}

\qquad

\noindent
{\bf Example} \ A $\PITL$ model $M_{\bf A}=\langle F_{\bf R},{\bf W}, P, I\rangle$ which is based on the real-time frame $F_{\bf R}$ and describes the working of a given probabilistic automaton ${\bf A}$ of the form (\ref{fpta}) from Definition \ref{pautomata} can be defined as follows. The vocabulary of $M_{\bf A}$ includes of the mandatory symbols $0$, $+$, $\ell$, \ldots , the transitions $a\in A$ as flexible $0$-ary predicate letters, and the choice probabilities $q_a$ as rigid constants. As for the duration probability densities $p_a$, it is convenient to have rigid unary function symbols $P_a$ which denote the functions $\lambda\tau.\int\limits_0^\tau p_b(t)dt$. The vocabulary does not provide direct reference to the {\em states} of ${\bf A}$ as done in $\PDC$; behaviour is instead described in terms of {\em transitions} whose beginnings and ends mark the times of state change. Every possible behaviour (\ref{beh}) is described by a $w\in{\bf W}$ such that  $I^w(a_i)\left(\left[\sum\limits_{j<i}d_j, \sum\limits_{j\leq i}d_j\right]\right)=1$. $I^w(a)([\tau_1,\tau_2])=1$ holds only if $[\tau_1,\tau_2]$ is one of the intervals $\left[\sum\limits_{j<i}d_j, \sum\limits_{j\leq i}d_j\right]$, $i<\omega$, and $a$ is the corresponding $a_i$. Given $w\in{\bf W}$ and $\tau\in{\bf R}_+$, $P^w(\tau,X)$ is defined as the probability for the finite behaviour described by $w$ up to time $\tau$ to develop into an infinite behaviour from $X$. For instance, let
\[\langle F_{\bf R},I^w\rangle,[0,\tau]\models(\top;a),\]
which means that the interval $[0,\tau]$ accommodates a finite sequence of transitions which ends at $a$ and a new transition is to begin at time $\tau$. Then, if $b\in A$ and $b^-=a^+$, $P^w$ satisfies the equality
\begin{equation}\label{exampleeq}
\textstyle P^w(\tau,\sem{(b\wedge x\leq\ell \wedge \ell\leq y;\top)}_{M_{\bf A},w,[\tau,\tau]})=q_b \int\limits_{I^w(x)}^{I^w(y)} p_b(t)dt.
\end{equation}
Here $\sem{(b\wedge x\leq\ell \wedge \ell\leq y;\top)}_{M_{\bf A},w,[\tau,\tau]}$ is the set of all the behaviours in which the part of $w$ until time $\tau$ is continued by transition $b$ and the duration of $b$ is in the range $[I^w(x),I^w(y)]$. The equality (\ref{exampleeq}) describes the probability for such a development to take place. If the source state of $b$ is $s_0$, then (\ref{exampleeq}) holds for $\tau=0$ and all $w$ as well. (\ref{exampleeq}) entails that the formula
\begin{equation}\label{exprobabformula}
\neg(\top;a;\ell = 0\wedge p((b\wedge x\leq\ell\wedge\ell\leq y;\top))\not= q_b.(P_b(y)-P_b(x))),
\end{equation}
is valid in $M_{\bf A}$. This formula means that the probability for a behaviour satisfying $(b\wedge x\leq\ell\wedge\ell\leq y;\top)$ to take place after $(\top;a)$ is $q_b.(P_b(y)-P_b(x))$, which, by the chosen interpretation of $P_b$, is equal to the righthand side of (\ref{exampleeq}).

Describing probabilistic real-time automata in a system of infinite interval probabilistic duration calculus which corresponds to $\PITL$ is the topic of Section \ref{modellingautomata}.  

We conclude the definition of $\PITL$ semantics with a remark on the underlying model of time. As mentioned in the introduction, $PDC$ and $\PITL$ are essentially branching-time interval logics. An alternative way to introduce the semantics of $\PITL$ could be to use partially ordered time domains $\langle T,\leq\rangle$ with some additional conditions on their maximal linearly ordered subsets. Given a $\PITL$ model $\langle F,{\bf W},I, P\rangle$ as described above, we can construct the corresponding partially ordered time domain by taking 
\[\{\langle\tau,{\bf W}_{w,\tau}\rangle:\tau\in T,w\in{\bf W}\}\]
as the set of time points and defining the partial ordering by the clause
\[\langle\tau_1,W_1\rangle\leq\langle \tau_2,W_2\rangle\mbox{ iff }\tau_1\leq\tau_2\mbox{ and }W_1\supseteq W_2.\]
The chosen way to define $\PITL$ models saves us the need to reformulate results on $\ITL$ which are essentially linear-time and are therefore known in the literature just for the sake of notation differences.

\section{A proof system for $\PITL$}
\label{proofsystem}

In this section we propose axioms and a proof rule for $\PITL$. If added to the complete proof system for $\ITL$ with infinite intervals from \cite{WX04} given in Section \ref{proofsystemitl}, these axioms and the rule form a system which is complete for $\PITL$ with respect to its abstract semantics introduced in Section \ref{pitlsemantics}. This is demonstrated in Section \ref{pitlcompleteness}. Most of our axioms and rule are modifications of those for $\PNL$ from \cite{Gue00probab}. The modifications were made to account for the use of infinite intervals instead of the $\NL$ expanding modalities. Some simple infinite-interval-specific properties of $p(.)$ are handled by completely new axioms. 

\subsection{The system}

\qquad

\qquad

\noindent{\em Extensionality}

\qquad

\begin{tabular}{lll}
$(P_;)$ & $(\ell=x;p(\psi)=y)\Rightarrow p((\ell=x;\psi))=y$\\
$(P_\infty)$ & $\ell=\infty\Rightarrow(\varphi\Leftrightarrow p(\varphi)=1)$\\
$(P_\leq)$ & $\displaystyle\frac{\vdash(\varphi;\ell=\infty)\Rightarrow(\psi\Rightarrow\chi)}{\vdash\varphi\wedge\ell<\infty\Rightarrow p(\psi)\leq p(\chi)}$\\
\end{tabular}

\qquad

\noindent{\em Arithmetics of probabilities}

\qquad

\qquad

\begin{tabular}{lll}
$(P_\bot)$ & $p(\bot)=0$\\
$(P_\top)$ & $p(\top)=1$\\
$(P_+)$ & $p(\varphi)+p(\psi)=
p(\varphi\vee\psi)+p(\varphi\wedge\psi)$\\
\end{tabular}

\qquad

$P_;$ expresses that the probability function $P_{\langle I,P\rangle,\max\sigma}$ which is used to evaluate $I_{\sigma}(p(\psi))$ depends on the end point $\max\sigma$ and not on the whole reference interval $\sigma$. $P_\infty$ means that having the entire future as the reference interval renders all properties deterministic: no alternative behaviours are possible "from $\infty$ on"; the interpretations $I'$ from $\langle I',P'\rangle\in{\bf W}_{\langle I,P\rangle,\infty}$ can differ from $I$ only on individual variables and such differences are disregarded in the definition (\ref{semvarphi}) of $\sem{\varphi}_{M,\langle I,P\rangle,\sigma}$ for all intervals $\sigma$. The rule $P_\leq$ means that if a property $\chi$ is a logical consequence of another property $\psi$, then the probability of $\chi$ is at least as big as that of $\psi$. The probabilities of $\psi$ and $\chi$ are compared in the context of a finite-interval condition $\varphi$. The case of an infinite-interval condition $\varphi$ is handled by axiom $P_\infty$. The axioms $P_\bot$, $P_\top$ and $P_+$ are self-explanatory. The correctness of the axioms and the rule is straightforward. The use of $\vdash$ in $P_\leq$ is to emphasize that we intend to apply this rule only to theorems. The maximal consistent sets of formulas which take part in our completeness argument for this proof system below need not be closed under $P_\leq$.

The rule $P_\leq$ can be classified under the category of probability arithmetics as well, because of the meaning of $\leq$, which is defined by (\ref{leqdef}). However, we find its role as an extensionality rule, which is further highlighted by the derived rule $\PITL 1$ below, to be more important.

\subsection{Some useful $\PITL$ theorems and a derived rule}
\label{usefultheorems}

The $\PITL$ theorems $\PITL 2$ and $\PITL 3$ and the derived rule $\PITL 1$ below are used in proofs in the rest of the paper. $\PITL 4$ is included to highlight the role of infinite intervals in the semantics of probability terms and the effect of $\tau$-equivalence on probabilities, respectively.

\qquad

\begin{tabular}{lll}
$(P_\leq^\infty)$ & $\displaystyle\frac{(\varphi;\ell=\infty)\vee(\varphi\wedge\ell=\infty)\Rightarrow(\psi\Rightarrow\chi)}{\varphi\Rightarrow p(\psi)\leq p(\chi)}$\\
$(\PITL 1)$ & $\displaystyle\frac{\varphi\Leftrightarrow\psi}{p(\varphi)=p(\psi)}$\\
$(\PITL 2)$ & $p(\varphi)+p(\neg\varphi)=1$\\
$(\PITL 3)$ & $p(\varphi)<p(\psi)\Rightarrow p(\psi\wedge\neg\varphi)\not=0$\\

$(\PITL 4)$ & $p(\varphi)=p(\varphi\wedge\ell=\infty)$\\
\end{tabular}

\qquad

\noindent
Here follows a derivation for $P_\leq^\infty$. The purely $\ITL$ parts are skipped and marked ``$\ITL$'' for the sake of brevity. Applications of the axioms $U1$-$U7$ for arithmetics on probability domains are skipped without comments.

\qquad

\begin{tabular}{rll}
1 & $(\varphi;\ell=\infty)\Rightarrow (\psi\Rightarrow\chi)$ & assumption, $\ITL$\\
2 & $\varphi\wedge\ell<\infty\Rightarrow p(\psi)\leq p(\chi)$ & 1, $P_\leq$\\
3 & $\ell=\infty\wedge\varphi\Rightarrow (p(\psi)=0\wedge p(\chi)=0)$ & assumption, $P_\infty$, $\PITL 2$\\ & $\qquad\qquad\qquad \vee (p(\psi)=0\wedge p(\chi)=1)$\\ & $\qquad\qquad\qquad \vee(p(\psi)=1\wedge p(\chi)=1)$ \\
4 & $ \varphi\wedge\ell=\infty\Rightarrow p(\psi)\leq p(\chi) $ & 3, $\ITL$\\
5 & $\ell<\infty\vee\ell=\infty$ & $\ITL$ \\
6 & $\varphi\Rightarrow p(\psi)\leq p(\chi)$ & 2, 4, 5\\
\end{tabular}

\qquad

\noindent
$\PITL 4$ is obtained by applying $P_\leq^\infty$ to the $\ITL$ theorems
\[\eqalign{
  (\top;\ell=\infty)\vee(\top\wedge\ell=\infty)
&\Rightarrow(\varphi\Rightarrow\varphi\wedge\ell=\infty)\mbox{ and }\cr
  (\top;\ell=\infty) \vee(\top\wedge\ell=\infty)
&\Rightarrow(\ell=\infty\wedge\varphi\Rightarrow\varphi).}\]
The rule $\PITL 1$ is proved by two applications of $P_\leq^\infty$ too. The proofs for $\PITL 2$ and $\PITL 3$ below are included as simple examples of the working of the axioms about arithmetics of probabilities. 

\qquad

\noindent
$\PITL 2$:

\qquad

\begin{tabular}{rll}
1 & $\varphi\wedge\neg\varphi\Leftrightarrow\bot$ & $\ITL$\\
2 & $p(\varphi\wedge\neg\varphi)=p(\bot)$ & 1, $\PITL 1$\\
3 & $p(\varphi\wedge\neg\varphi)=0$ & 2, $P_\bot$\\
4 & $\varphi\vee\neg\varphi\Leftrightarrow\top$ & $\ITL$\\
5 & $p(\varphi\vee\neg\varphi)=p(\top)$ & 4, $\PITL 1$\\
6 & $p(\varphi\wedge\neg\varphi)=1$ & 5, $P_\top$\\
7 & $p(\varphi)+p(\neg\varphi)=p(\varphi\wedge\neg\varphi)+p(\varphi\wedge\neg\varphi)$ & $P_+$\\
8 & $p(\varphi)+p(\neg\varphi)=1$ & 2, 6, 7, $\ITL$\\
\end{tabular}

\qquad

\noindent
$\PITL 3$:

\qquad

\begin{tabular}{rll}
1 & $p(\psi)\leq p(\varphi\vee\psi)$ & $P_\leq^\infty$\\
2 & $p(\varphi)+p(\psi\wedge\neg\varphi)=p(\varphi\wedge\psi\wedge\neg\varphi)+p(\varphi\vee\psi\wedge\neg\varphi)$ & $P_+$\\
3 & $p(\varphi)+p(\psi\wedge\neg\varphi)=p(\varphi\vee\psi)$ & 2, $\PITL 1$, $P_\bot$\\
4 & $p(\varphi)<p(\psi)\Rightarrow p(\varphi)<p(\varphi\vee\psi)$ & 1\\
5 & $p(\varphi)<p(\psi)\Rightarrow p(\psi\wedge\neg\varphi)\not=0$ & 3, 4\\
\end{tabular}

\section{Completeness of the proof system for $\PITL$}
\label{pitlcompleteness}

In this section we show that the proof system for $\PITL$ from Section \ref{proofsystem} is complete. To exploit the full potential of the abstract semantics of $\PITL$, we prove a strong completeness theorem. It states that every consistent set of $\PITL$ formulas has a model. This is convenient for the study of further extensions of the logic whose syntactic elements can be represented by adding infinitely many non-logical symbols and axioms about them, or when a modelled system is described using infinitely many formulas.  

The main step in this proof is the construction of what is known in model theory as the {\em elementary diagram} $\Delta$ of a $\PITL$ model $M$ for an arbitrary given set of $\PITL$ formulas $\Gamma$ which is consistent in the proposed proof system for $\PITL$. $\Delta$ is a description of $M$ in a $\PITL$ language whose vocabulary has names for all the elements of $M$. To avoid repeating the technical steps which are not specific to the probability operator of $\PITL$ and can be found in the completeness proof for (non-probabilistic) $\ITL$ with infinite intervals from \cite{WX04}, we introduce a translation of the involved $\PITL$ languages into corresponding $\ITL$ languages with appropriate vocabularies and use it to view subsets of the constructed diagram and the whole diagram as complete Henkin theories in (non-probabilistic) $\ITL$ as well.

The model $M$ that we construct is very similar to a canonical model. We stop short of calling it canonical, because of the dedicated technique which is used to build the behaviour representations $v$ which are needed to populate the sets $\sem{\varphi}_{M,w,\sigma}$ for $\varphi$, $\sigma$ and $w$ such that  $M,w,\sigma\models p(\varphi)\not=0$ is supposed to hold.

Without losing generality, we consider only sets of formulas $\Gamma$ which contain $\ell=\infty$. This way we restrict ourselves to seeking the satisfaction of $\Gamma$ at an infinite interval. The satisfaction of a consistent $\Gamma$ which is not consistent with $\ell=\infty$ can be achieved through the satisfaction of
\begin{equation}\label{infiniteell}
\{\ell=\infty\}\cup\{(\gamma\wedge\ell=c;\top):\gamma\in\Gamma\}
\end{equation}
where $c$ is some fresh rigid constant.

The completeness argument involves the application of some non-trivial results about interpolation in $\ITL$. We present them first.

\subsection{Interval-related and Craig interpolation in $\ITL$ with infinite intervals}
\label{interpolation}

Inter\-val-related interpolation for $\ITL$ with finite intervals,
$\NL$ and a subset of $\DC$ with finite intervals and projection onto
state were formulated and proved in \cite{Gue01,Gue04b}. Craig
interpolation was shown to hold for these logics there too. Here we
just formulate interval-related interpolation for $\ITL$ with infinite
intervals in the special form which is convenient for our completeness
argument.

Let ${\bf L}$ and ${\bf L}'$ be two vocabularies for $\ITL$ with infinite intervals. Let ${\bf L}$ and ${\bf L}'$ share their rigid symbols, including the individual variables, and let the only flexible symbol occurring in both ${\bf L}$ and ${\bf L}'$ be $\ell$. Let there be a bijection between the flexible symbols from ${\bf L}\setminus\{\ell\}$ and those from ${\bf L}'$ such that the symbol $s'$ from ${\bf L}'$ which corresponds to $s\in{\bf L}$ is of the same kind and arity as $s$. Let $\varphi'$ denote the result of replacing each flexible symbol $s\in{\bf L}\setminus\{\ell\}$ in a formula $\varphi$ written in ${\bf L}$ by the corresponding $s'\in{\bf L}'$.

\begin{thm}
\label{intervalrelatedinterpolation}
Let $\Phi$ be a finite set of formulas and $\varphi$ and $\psi$ be two more formulas, all written in ${\bf L}$. Let $c$ be a rigid constant in ${\bf L}$. Let 
\[\left(\ell=c\wedge\Box\forall\bigwedge\limits_{\chi\in\Phi}(\chi\Leftrightarrow\chi');\ell=\infty\right)\Rightarrow(\varphi\Rightarrow\psi')\]\
be theorem of $\ITL$ with infinite intervals. Then there is a formula $\theta$ written in ${\bf L}$ such that
\[\varphi\wedge c<\infty\wedge\ell=\infty\Rightarrow(\ell=c\wedge\theta;\ell=\infty)\mbox{ and }(\ell=c\wedge\theta';\ell=\infty)\Rightarrow\psi'\]
are theorems of $\ITL$ as well.
\end{thm}

We use the standard form of Craig interpolation:

\begin{thm}
\label{craiginterpolation}
Let ${\bf L}_1$ and ${\bf L}_2$ be two $\ITL$ vocabularies. Let $\varphi_i$ be a formula of $\ITL$ with infinite intervals written in the vocabulary ${\bf L}_i$, $i=1,2$, and
\[\varphi_1\Rightarrow\varphi_2\]
be a theorem of $\ITL$ with infinite intervals.
Then there is a formula $\theta$ written in the vocabulary ${\bf L}_1\cap{\bf L}_2$ such that both
\[\varphi_1\Rightarrow\theta\mbox{ and }\theta\Rightarrow\varphi_2\]
are such theorems.
\end{thm}

The proofs of the two interpolation theorems are simple variants of those of the theorems known from \cite{Gue01}, which in their turn follow the pattern of the model-theoretic proof of Craig interpolation that can be seen in, e.g., \cite{ChK73}.

\subsection{Consistency in $\PITL$}

\begin{defi}
Given an $\ITL$ ($\PITL$) vocabulary ${\bf L}$, $\ITL_{\bf L}$ ($\PITL_{\bf L}$) denotes the set of the theorems of $\ITL$ ($\PITL$) written in a given vocabulary ${\bf L}$. Given ${\bf L}$ and a set of formulas $\Gamma$ written in ${\bf L}$, $\Cn_{{\bf L},\ITL}(\Gamma)$ ($\Cn_{{\bf L},\PITL}(\Gamma)$) denotes the set of formulas written in ${\bf L}$ which can be proved using formulas from  $\ITL_{\bf L}\cup\Gamma$ ($\PITL_{\bf L}\cup\Gamma$) and the propositional logic rule {\em Modus Ponens} $\varphi,\ \varphi\Rightarrow\psi\,/\, \psi$.
\end{defi}

\begin{defi}
A set of $\ITL$ ($\PITL$) formulas $\Gamma$ written in a vocabulary ${\bf L}$ is {\em consistent} if $\bot\not\in\Cn_{{\bf L},\ITL}(\Gamma)$ ($\bot\not\in\Cn_{{\bf L},\PITL}(\Gamma)$). A consistent  $\Gamma$ is {\em maximal in ${\bf L}$} if it has no consistent proper supersets of formulas written in ${\bf L}$. 

Just like in first-order predicate logic, a set of formulas $\Gamma$ {\em has witnesses in some set of rigid constants $C$} if for every existential formula $\exists x\varphi\in\Gamma$ there is a {\em witness} $c\in C$ such that $[c/x]\varphi\in\Gamma$. 
\end{defi}

Here follows the {\em Lindenbaum Lemma} for $\PITL$ as known from numerous predicate and modal logics:

\begin{thm}
\label{lindenbaum}
Let $\Gamma$ be a consistent set of formulas $\PITL$ written in some vocabulary ${\bf L}$ and $C$ be a countably-infinite set which consists of infinitely many fresh constants of both the sort of durations and the sort of probabilities. Then there is a maximal consistent set of formulas written in ${\bf L}\cup C$ which contains $\Gamma$ and has witnesses in $C$.
\end{thm}

We omit the proof for $\PITL$, because it is the same as that for $\ITL$ with abstract semantics and finite intervals which can be seen in \cite{Dut95}. The proof for $\ITL$ with infinite intervals was omitted in \cite{WX04} for the same reason.

\subsection{A vocabulary for the elementary diagram $\Delta$ for the $\PITL$ model $M$}

The $\PITL$ vocabulary ${\bf L}_D$ which we introduce next is structured so that a $\PITL$ model $M$ for the extension of some given $\PITL$ vocabulary ${\bf L}$ by a countable set of fresh rigid constants that we construct below can be fully described in it in terms of rather simple quantifier- and variable-free formulas which can be regarded as making up a {\em diagram} $\Delta$ for $M$ in the model-theoretic sense. ${\bf L}_D$ contains rigid constants to name all the elements of the duration domain and the probability domain of $M$ and a separate set of flexible symbols to describe the behaviour of the flexible symbols of ${\bf L}$ in each interpretation from $M$. Indeed, we construct an {\em elementary diagram} for $M$ in ${\bf L}_D$, which consists of all the formulas in ${\bf L}_D$ which hold at some infinite interval in $M$ under the convention that formulas written in the various sets of flexible symbols mentioned above are understood to hold at the respective interpretations. 

${\bf L}_D$ is the union of the following sets of symbols:

1. The rigid symbols of ${\bf L}$, including the individual variables, and the mandatory flexible constant $\ell$.

2. Two countably-infinite sets of fresh rigid constants $C^d$ and $C^p$ of the sorts of durations and probabilities, respectively, whose structure is explained below. 

3. The fresh flexible symbols $s^\nu$, $\nu\in S$, of the same kind and arity as $s$, for each flexible $s\in{\bf L}\setminus\{\ell\}$. The countably-infinite index set $S$ is defined below. 

$C^d$ and $C^p$ are assumed to be the countably-infinite disjoint unions of some countably infinite sets $C^d_k$ and $C^p_k$, $k<\omega$, respectively. Similarly, $S$ is assumed to be the countably-infinite union of the sets $S_k$, $k<\omega$.
We denote $\bigcup\limits_{i\leq k}C^d_i$, $\bigcup\limits_{i\leq k}C^p_i$ and $\bigcup\limits_{i\leq k}S_i$ by $C^d_{\leq k}$, $C^p_{\leq k}$ and $S_{\leq k}$, respectively, for all $k<\omega$. We denote the vocabulary which consists of the rigid symbols of ${\bf L}$, $\ell$, the rigid constants from $C^d_{\leq k}$ and $C^p_{\leq k}$ and the flexible symbols $s^\nu$ for $\nu\in S_{\leq k}$ by ${\bf L}_{\leq k}$ for all $k<\omega$. We denote the extension of ${\bf L}_{\leq k}$ by the flexible symbols $s^\nu$ for $\nu\in S_{\leq k+1}$ by ${\bf L}_{\leq k+1}'$.

The set $S_0$ is the singleton $\{\langle\rangle\}$, which consists of the empty list $\langle\rangle$. 
\[S_{k+1}=\{\langle \nu,c,\varphi\rangle:\nu\in S_{\leq k},c\in C^d_{\leq k},\varphi\mbox{ is written in }{\bf L}_{\leq k}\}\mbox{ for all }k<\omega.\]

In the construction of $\Delta$ below, given a $\nu\in S$, $A^\nu$ stands for the result of replacing the flexible symbols $s\in{\bf L}\setminus\{\ell\}$ in a term or formula $A$ written in the vocabulary ${\bf L}\cup C^d\cup C^p$ by their corresponding symbols $s^\nu$. We denote the vocabulary which consists of the rigid symbols of ${\bf L}$, including the individual variables, $\ell$ and the flexible symbols $s^\nu$ for some fixed $\nu\in S$ and all flexible $s\in{\bf L}\setminus\{\ell\}$ by ${\bf L}^\nu$.

\subsection{A translation of $\PITL$ formulas into $\ITL$}

Let ${\bf L}$ be a $\PITL$ vocabulary. We define its corresponding
vocabulary ${\bf L}_\ITL$ for two-sorted (non-probabilistic) $\ITL$
with infinite intervals with the sorts of durations and probabilities
as in $\PITL$. Roughly speaking, ${\bf L}_\ITL$ is an extension of
${\bf L}$ by flexible constants and function symbols which are meant
to simulate probability terms. Here follows the precise definition.

\begin{defi}
\label{itlvoc}
${\bf L}_\ITL$ is the union of the vocabularies ${\bf L}_{\ITL,k}$,
$k<\omega$. ${\bf L}_{\ITL,0}$ is ${\bf L}$. Given ${\bf L}_{\ITL,i}$,
$i\leq k$, ${\bf L}_{\ITL,k+1}$ is the set of flexible constants and
function symbols
\[\{\probab_\varphi:\varphi\mbox{ is a formula written in }\bigcup\limits_{i\leq k}{\bf L}_{\ITL,k}\mbox{ and contains at least one symbol from }{\bf L}_{\ITL,k}\}.\]
The values of the symbols $\probab_\varphi$ are of the probability
sort. If $\varphi$ has no free variables, then $\probab_\varphi$ is a
flexible constant. Otherwise $\probab_\varphi$ is a flexible function
symbol whose arity is $|FV(\varphi)|$ and the sort of the $i$th
argument of $\probab_\varphi$ is that of the $i$th free variable of
$\varphi$ with respect to some fixed ordering of these variables,
$i=1,\ldots,|FV(\varphi)|$.
\end{defi}

Next we define a translation $\tr$ of $\PITL$ terms and formulas
written in ${\bf L}$ into $\ITL$ formulas written in ${\bf
L}_\ITL$. The goal of $\tr$ is to systematically replace the
occurrences of probability terms by terms built using the
corresponding constant and function symbols from Definition
\ref{itlvoc}. To achieve this, $\tr$ works by the following rule:
\begin{equation}
\label{translationA}
[p(\psi_1)/z_1,\ldots,p(\psi_n)/z_n]A
\end{equation}
where denotes $A$ a term or formula with no probability terms is translated into
\begin{equation}
\label{translationA2}
[\probab_{\tr(\psi_1)}(x_{1,1},\ldots,x_{1,m_1})/z_1,\ldots,\probab_{\tr(\psi_1)}(x_{n,1},\ldots,x_{n,m_n})/z_n]A
\end{equation}
where $x_{i,1}$, \ldots , $x_{i,m_i}$ are the free variables of $\psi_i$ in the fixed ordering mentioned above, $i=1,\ldots,n$. If $FV(\psi)=\emptyset$, then the expression $\probab_{\tr(\psi_i)}(x_{i,1},\ldots,x_{i,m_i})$ denotes just the flexible constant $\probab_{\tr(\psi_i)}$. 

\qquad

\noindent
{\bf Example} \ If there are no probability terms in $\varphi$ and $FV(\varphi)=x_1$, then $\tr(p(\varphi))$ is the term $\probab_{\varphi}(x_1)$ and $\tr(p((\ell=x_2;p(\varphi)<p(\neg\varphi))))$ is
$\probab_{(\ell=x_2;\probab_{\varphi}(x_1)<\probab_{\neg\varphi}(x_1))}(x_1,x_2)$.

\qquad

Every term and formula can be represented in the form (\ref{translationA}) in a unique way up to renaming the distinct variables $z_1,\ldots,z_n$, if we assume that all of these variables have free occurrences in $A$ and that the formulas $\psi_1,\ldots,\psi_n$ are all different. The semantical correctness of the substitution in (\ref{translationA}) and (\ref{translationA2}) is not relevant to this definition of $\tr$. Given a set of $\PITL$ formulas $\Gamma$, we denote $\{\tr(\gamma):\gamma\in\Gamma\}$ by $\tr(\Gamma)$.

Terms built using the function symbols $\probab_\psi$ from ${\bf L}_\ITL$ in translations of $\PITL$ formulas always have the free variables of $\psi$ as their argument terms. That is why formulas written in ${\bf L}_\ITL$ which contain $\probab_\psi$ in terms of other forms are not in the range of $\tr$. However, they always have equivalents of the form $\tr(\varphi)$ for appropriate $\PITL$ formulas $\varphi$ written in ${\bf L}$. To realise that, note that if $FV(\psi)=\{x_1,\ldots,x_n\}$ and $y_1,\ldots,y_n$ are $n$ fresh variables of the appropriate sorts, then $\probab_\psi(t_1,\ldots,t_n)=z$ is equivalent to
\[\exists y_1\ldots\exists y_n\left(\bigwedge\limits_{i=1}^n t_i=y_i\wedge\exists x_1\ldots\exists x_n\left(\bigwedge\limits_{i=1}^n y_i=x_i\wedge\probab_\psi(x_1,\ldots,x_n)=z\right)\right).\]
Furthermore, every formula written in ${\bf L}_\ITL$ has an equivalent in which the terms of the form $\probab_\psi(t_1,\ldots,t_n)$ appear only in atomic formulas of the form $\probab_\psi(t_1,\ldots,t_n)=z$ where $z$ can be chosen to be different from $x_1,\ldots,x_n$. 

Now we turn to the correspondence between derivability in $\PITL$ and $\ITL$ with infinite intervals.

\begin{prop}
\label{consistencyundertranslation}
Let ${\bf L}$ be a $\PITL$ vocabulary and $\Gamma$ be a set of formulas written in ${\bf L}$. Then
\[\tr(\Cn_{{\bf L},\PITL}(\Gamma))=\Cn_{{\bf L}_\ITL,\ITL}(\tr(\PITL_{\bf L}\cup\Gamma)).\]
\end{prop}

\begin{proof}
Simple induction on the construction of proofs.
\end{proof}

\begin{cor}
A set of $\PITL$ formulas $\Gamma$ written in a vocabulary ${\bf L}$ is  consistent iff\\ $\Cn_{{\bf L}_\ITL,\ITL}(\PITL_{\bf L}\cup\Gamma)$ is consistent.

\end{cor}

\begin{proof}
$\tr(\bot)$ is $\bot$.
\end{proof}

\subsection{The weakened proof system $\PITL^-$}

The model $M$ constructed below is for ${\bf L}\cup C^d\cup C^p$. It contains one class of $w\in{\bf W}$ which are the same except possibly for the interpretations $I^w$ of some individual variables for every $\nu\in S$. Let $w_\nu$ denote a representative for the class of interpretations corresponding to $\nu$. Then $I^{w_\nu}(s)$ is defined by the formulas from the diagram $\Delta$ for $M$ which describe $s^\nu$ for all flexible $s\in{\bf L}\setminus\{\ell\}$. We are interested in having a set of formulas $\Gamma$ which contains the formula $\ell=\infty$ satisfied at some infinite interval $[\tau_0,\infty]$ and some interpretation $I$ in $M$. Our construction of $M$ provides that if $c\in C^d$ and $\tau_1$ is defined by the equality $m([\tau_0,\tau_1])=I^{w_\nu}(c)$ in $M$, then $w_\nu$ and $w_{\langle\nu,c,\varphi\rangle}$ are related as follows:

\begin{quote}
If $M,w_\nu,[\tau_0,\tau_1]\models p(\varphi)\not=0$ and $FV(\varphi)=\{x_1,\ldots,x_n\}$, then $w_\nu\equiv_{\tau_1}w_{\langle\nu,c,\varphi\rangle}$ and $M,v,[\tau_0,\infty]\models\varphi$ for some $v$ such that $I^v=(I^{w_{\langle \nu,c,\varphi\rangle}})^{I_\nu(x_1),\ldots,I_\nu(x_n)}_{x_1\ ,\ \ldots\ ,\ x_n}$ and $P^v=P^{w_{\langle \nu,c,\varphi\rangle}}$.
\end{quote}
This means that $w_{\langle\nu,c,\varphi\rangle}\in\sem{\varphi}_{M,w_\nu,[\tau_0,\tau_1]}$.

Furthermore, we are interested in enforcing $\PITL$ local logical consequence at each particular $w\in{\bf W}$, but not across different $w$. That is why in the construction of $\Delta$ below we restrict the applicability of the $\PITL$-specific axioms $P_;$, $P_\infty$, $P_\bot$, $P_\top$ and $P_+$ and rule $P_\leq$ from Section \ref{proofsystem} in sets of formulas written in ${\bf L}_D$. We allow only instances of $P_;$, $P_\infty$, $P_\leq$, $P_\bot$, $P_\top$ and $P_+$ in which all flexible symbols except $\ell$ have the same superscript $\nu\in S$. The resulting weakened proof system is tied to the vocabulary ${\bf L}_D$. We denote it and the set of its theorems written in a given sub-vocabulary ${\bf L}'$ of ${\bf L}_D$ by $\PITL^-$ and  $\PITL^-_{{\bf L}'}$, respectively. Theorem \ref{lindenbaum} applies to consistency with $\PITL^-_{{\bf L}'}$ without change. Similarly, we have the following variant of Proposition \ref{consistencyundertranslation}:

\begin{prop}
\label{consistencyundertranslation2}
Let ${\bf L}'$ be a sub-vocabulary of ${\bf L}_D$ and $\Gamma$ be a set of formulas written in ${\bf L}'$. Then
\[\tr(\Cn_{{\bf L}',\PITL^-}(\Gamma))=\Cn_{{\bf L}'_\ITL,\ITL}(\tr(\PITL^-_{{\bf L}'}\cup\Gamma)).\]
\end{prop}

We also use the following somewhat more involved technical consequence
of the restricted use of the instances of $P_;$, $P_\infty$, $P_\leq$,
$P_\bot$, $P_\top$ and $P_+$ and the restricted application of
$P_\leq$.

\begin{lem}
\label{seppitlminus}
Let $\alpha\in\PITL^-_{{\bf L}'}$ for some sub-vocabulary ${\bf L}'$ of ${\bf L}_D$. Let $C$ be the set of the rigid constants of ${\bf L}'$. Then there exist finitely many superscripts $\nu_1,\ldots,\nu_n\in S$ and theorems $\beta_i\in\PITL_{{\bf L}^{\nu_i}\cup C}$, $i=1,\ldots,n$, such that the formula
\begin{equation}\label{seppitlminusformula}
\bigwedge\limits_{i=1}^n\Box\forall\beta_i\Rightarrow\alpha
\end{equation}
is provable without the use of $P_;$, $P_\infty$, $P_\bot$, $P_\top$ and $P_+$ and $P_\leq$, that is, essentially in (non-probabilistic) $\ITL$ with infinite intervals.
\end{lem}

\begin{proof}
Consider a $\PITL^-$ proof of $\alpha$ in ${\bf L}'$. Let $\nu_1,\ldots,\nu_n$ be all the superscripts of flexible symbols occurring in formulas from this proof. If a formula $\beta$ from the proof is written in the vocabulary ${\bf L}^{\nu_i}\cup C$ for some $i\in\{1,\ldots,n\}$, then $\beta\in\PITL_{{\bf L}^{\nu_i}\cup C}$. To realise this, notice that changing all the superscripts of the flexible symbols in the formulas from the part of the proof which leads to $\beta$ to $\nu_i$ preserves its correctness. We can choose $\beta_i$ to be the conjunction of all the formulas from $\PITL_{{\bf L}^{\nu_i}\cup C}$ in the chosen proof of $\alpha$, $i=1,\ldots,n$.
\end{proof}

Consistency in the rest of this section is with respect to $\PITL^-$.

\subsection{The elementary diagram $\Delta$ for $M$}

Here follows the precise construction of the diagram $\Delta$. 

$\Delta$ is the union of the infinite ascending sequence of sets of formulas 
\begin{equation}\label{deltasequence}
\Delta_0\subset\Delta_1'\subset\Delta_1\subset\ldots\subset\Delta_k'\subset\Delta_k\subset\ldots
\end{equation}
where $\Delta_k$ and $\Delta_{k+1}'$ consist of formulas written in ${\bf L}_{\leq k}$ and ${\bf L}_{\leq k+1}'$, respectively, for each $k<\omega$. $\Delta_0$ is a maximal consistent set with witnesses in $C^d_0\cup C^p_0$ which contains the set $\{\gamma^{\langle\rangle}:\gamma\in\Gamma\}$. Such a set exists by Theorem \ref{lindenbaum}. For an arbitrary $k<\omega$, $\Delta_{k+1}'$ is the extension of $\Delta_k$ by
\begin{equation}
\label{addedformulas}
\mbox{the formula }\varphi^{\nu'}\mbox{ and the formulas }(\Box\forall(\chi^\nu\Leftrightarrow\chi^{\nu'})\wedge\ell=c;\ell=\infty)\mbox{
for all $\chi$ written in }{\bf L},
\end{equation}
for each pair of indices $\nu\in S_{\leq k}$ and $\nu'\in S_{k+1}$ such that $\nu'=\langle\nu,c,\varphi\rangle$ and\\ $(p(\varphi^\nu)\not=0\wedge\ell=c;\ell=\infty)\in\Delta_k$.

\begin{lem}
\label{greatlemma}
If $\Delta_k$ is consistent, then $\Delta_{k+1}'$ is consistent too.
\end{lem}

The proof of this lemma is the key technical step in the entire completeness argument about our proof system for $\PITL$.

\begin{proof}
Assume that $\Delta_k$ is consistent and $\Delta_{k+1}'$ is not for the sake of contradiction. Since proofs in $\PITL^-$ are finitary, there is a finite inconsistent $\Xi\subset\Delta_{k+1}'$. $\Xi\not\subseteq\Delta_k$, because $\Delta_k$ is a consistent set. Hence there are finitely many $\nu'\in S_{k+1}\setminus S_{\leq k}$ such that flexible symbols superscripted by $\nu'$ occur in formulas from $\Xi$. These formulas are of some of the forms (\ref{addedformulas}). Below we prove that the assumed inconsistency of $\Xi$ is preserved after withdrawing the formulas of the forms (\ref{addedformulas}) for each such $\nu'\in S_{k+1}\setminus S_{\leq k}$. The remaining formulas in $\Xi$ are also in $\Delta_k$. This will bring contradiction with the assumed consistency of $\Delta_k$. Let us choose one such $\nu'$ and let $\nu'=\langle\nu,c,\varphi\rangle$. This means that $(p(\varphi^\nu)\not=0\wedge\ell=c;\ell=\infty)\in\Delta_k$. Then the formulas (\ref{addedformulas}) for the chosen $\nu'$ and $\nu$ are in $\Delta_{k+1}'$. Let the formulas in $\Xi$ with flexible symbols superscripted by $\nu'$ be $(\Box\forall(\chi_i^\nu\Leftrightarrow\chi_i^{\nu'})\wedge\ell=c;\ell=\infty)$, $i=1,\ldots,m$, and $\varphi^{\nu'}$. Let $\Xi_{\overline{\nu'}}$ be the set of the remaining formulas from $\Xi$, which have no flexible symbols superscripted by $\nu'$. Then
\[\vdash_{\PITL^-_{{\bf L}_{\leq k+1}'}}(\bigwedge\Xi_{\overline{\nu'}})\Rightarrow\left(\bigwedge\limits_{i=1}^m(\Box\forall(\chi_i^\nu\Leftrightarrow\chi_i^{\nu'})\wedge\ell=c;\ell=\infty)\Rightarrow\neg\varphi^{\nu'}\right).\]
Now Proposition \ref{consistencyundertranslation2} entails that
\[\vdash_{\ITL}\tr(\alpha)\Rightarrow\left(\tr(\bigwedge\Xi_{\overline{\nu'}})\Rightarrow\left(\bigwedge\limits_{i=1}^m(\Box\forall(\tr(\chi_i^\nu)\Leftrightarrow\tr(\chi_i^{\nu'}))\wedge\ell=c;\ell=\infty)\Rightarrow\neg\tr(\varphi^{\nu'})\right)\right)\]
where $\alpha\in \PITL^-_{{\bf L}_{\leq k+1}'}$. According to Lemma \ref{seppitlminus}, there is a finite set of superscripts $\nu_1,\ldots,\nu_n\in S_{\leq k+1}$ and this many formulas $\beta_i\in\PITL_{{\bf L}^{\nu_i}\cup C^d_{\leq k}\cup C^p_{\leq k}}$, $i=1,\ldots,n$, such that (\ref{seppitlminusformula}) is provable without the $\PITL$-specific axioms and rule, that is, essentially in $\ITL$ with infinite intervals. Without loss of generality we can assume that $\beta_1\in\PITL_{{\bf L}^{\nu}\cup C^d_{\leq k}\cup C^p_{\leq k}}$ and $\beta_2\in\PITL_{{\bf L}^{\nu'}\cup C^d_{\leq k}\cup C^p_{\leq k}}$. Then we have  
\[
\begin{array}{l}\vdash_{\ITL}\tr\left((\bigwedge\Xi_{\overline{\nu'}})\wedge\bigwedge\limits_{i=3}^n\Box\forall\beta_i\right)\Rightarrow\\
\qquad\left(\tr(\Box\forall\beta_1)\wedge\tr(\Box\forall\beta_2)\wedge\bigwedge\limits_{i=1}^m(\Box\forall(\tr(\chi_i^\nu)\Leftrightarrow\tr(\chi_i^{\nu'}))\wedge\ell=c;\ell=\infty)\Rightarrow\neg\tr(\varphi^{\nu'})\right).\end{array}\]
All the flexible symbols on the right of the main $\Rightarrow$ in this formula except $\ell$ are superscripted by either $\nu$ or $\nu'$ and the superscript $\nu'$ does not appear on symbols in the formula on the left of $\Rightarrow$. Hence by Craig interpolation (Theorem \ref{craiginterpolation}) some $\ITL$ formula $\lambda$ written in $({\bf L}^{\nu}\cup C^d_{\leq k}\cup C^p_{\leq k})_\ITL$ satisfies both
\begin{equation}\label{craig1}
\vdash_{\ITL}\tr\left((\bigwedge\Xi_{\overline{\nu'}})\wedge\bigwedge\limits_{i=3}^n\Box\forall\beta_i\right)\Rightarrow\lambda
\end{equation}
and
\begin{equation}\label{craig2}
\vdash_{\ITL}\bigwedge\limits_{i=1}^m(\Box\forall(\tr(\chi_i^\nu)\Leftrightarrow\tr(\chi_i^{\nu'}))\wedge\ell=c;\ell=\infty)\Rightarrow((\lambda\wedge\tr(\Box\forall\beta_1))\Rightarrow(\tr(\Box\forall\beta_2)\Rightarrow\neg\tr(\varphi^{\nu'}))).
\end{equation}
The formulas $\lambda\wedge\tr(\Box\forall\beta_1)$ and $\tr(\Box\forall\beta_2)\Rightarrow\neg\tr(\varphi^{\nu'})$ in (\ref{craig2}) are written in $({\bf L}^\nu \cup C^d_{\leq k}\cup C^p_{\leq k})_\ITL$ and $({\bf L}^{\nu'}\cup C^d_{\leq k}\cup C^p_{\leq k})_\ITL$, respectively. A bijection can be defined between the sets of the flexible symbols of these two vocabularies, excluding $\ell$, in which the flexible symbol $s'\in({\bf L}^{\nu'}\cup C^d_{\leq k}\cup C^p_{\leq k})_\ITL\setminus\{\ell\}$ which corresponds to $s\in({\bf L}^\nu \cup C^d_{\leq k}\cup C^p_{\leq k})_\ITL\setminus\{\ell\}$ is obtained by changing all the superscripts $\nu$ in $s$ to $\nu'$ and vice-versa. If $s$ is of the form $\probab_{\tr(\psi)}$ (see Definition \ref{itlvoc}), it may have more than one occurrence of a superscript $\nu$ in the subscript formula $\tr(\psi)$. All these occurrences have to be changed.
This bijection allows us to apply 
interval-related interpolation (Theorem \ref{intervalrelatedinterpolation}) to (\ref{craig2}) and conclude that some $\ITL$ formulas $\theta_\ITL\in({\bf L}^\nu\cup C^d_{\leq k}\cup C^p_{\leq k})_\ITL$ and $\theta_\ITL'\in({\bf L}^{\nu'}\cup C^d_{\leq k}\cup C^p_{\leq k})_\ITL$ which can be obtained from each other by replacing the corresponding flexible symbols from their respective vocabularies satisfy
\begin{equation}\label{interval1}
\vdash_{\ITL}\lambda\wedge\tr(\Box\forall\beta_1)\wedge c<\infty\wedge\ell=\infty\Rightarrow(\ell=c\wedge\theta_\ITL;\ell=\infty)
\end{equation}
and
\[\vdash_{\ITL}(\ell=c\wedge\theta_\ITL';\ell=\infty)\Rightarrow(\tr(\Box\forall\beta_2)\Rightarrow\neg\tr(\varphi^{\nu'}))\]
which by simply changing all superscripts $\nu'$ to $\nu$ implies
\begin{equation}\label{interval2}
\vdash_{\ITL}(\ell=c\wedge\theta_\ITL;\ell=\infty)\Rightarrow(\tr(\Box\forall\beta_2')\Rightarrow\neg\tr(\varphi^{\nu}))
\end{equation}
where $\beta_2'$ is the result of changing all the superscripts $\nu'$ of the flexible symbols in $\beta_2$ to $\nu$. By (\ref{craig1}) and (\ref{interval1}) we obtain 
\begin{equation}\label{cr1int1}
\vdash_{\ITL}\tr\left((\bigwedge\Xi_{\overline{\nu'}})\wedge\bigwedge\limits_{i=3}^n\Box\forall\beta_i\right)\wedge\tr(\Box\forall\beta_1)\wedge c<\infty\wedge\ell=\infty\Rightarrow(\ell=c\wedge\theta_\ITL;\ell=\infty)
\end{equation}
The formula $\theta_\ITL$ is the $\tr$-translation of some $\PITL$ formula written in ${\bf L}^\nu\cup C^d_{\leq k}\cup C^p_{\leq k}$ which, in its turn, has the form $\theta^\nu$ where $\theta$ is a formula written in ${\bf L}\cup C^d_{\leq k}\cup C^p_{\leq k}$. (Then $\theta_\ITL'$ is $\tr(\theta^{\nu'})$.)
Hence we have  
\[\vdash_{\PITL^-_{{\bf L}_{\leq k+1}'}}(\bigwedge\Xi_{\overline{\nu'}})\wedge\bigwedge\limits_{i=3}^n\Box\forall\beta_i\wedge\Box\forall\beta_1\wedge c<\infty\wedge\ell=\infty\Rightarrow(\ell=c\wedge\theta^\nu;\ell=\infty).\]
Since $\beta_i\in\PITL_{{\bf L}^{\nu_i}\cup C^d_{\leq k}\cup C^p_{\leq k}}\subseteq \PITL^-_{{\bf L}_{\leq k+1}'}$, $i=3,\ldots,n$, and $\beta_1\in\PITL_{{\bf L}^{\nu}\cup C^d_{\leq k}\cup C^p_{\leq k}}\subseteq \PITL^-_{{\bf L}_{\leq k+1}'}$, the above formula can be simplified to
\[\vdash_{\PITL^-_{{\bf L}_{\leq k+1}'}}(\bigwedge\Xi_{\overline{\nu'}})\wedge c<\infty\wedge\ell=\infty\Rightarrow(\ell=c\wedge\theta^\nu;\ell=\infty).\]
Since $(p(\varphi^\nu)\not=0\wedge\ell=c;\ell=\infty)\in\Delta_k$, $c<\infty,\ell=\infty\in\Delta_k$ too. This implies that $(\ell=c\wedge\theta^\nu;\ell=\infty)\in \Cn_{{\bf L}_{\leq k+1}'}(\Delta_k\cup\Xi_{\overline{\nu'}})$. Similarly, (\ref{interval2}) implies than 
\[\vdash_{\PITL_{{\bf L}^\nu\cup C^d_{\leq k}\cup C^p_{\leq k}}}(\ell=c\wedge\theta^\nu;\ell=\infty)\Rightarrow(\Box\forall\beta_2'\Rightarrow\neg\varphi^{\nu}),
\]
and, since $\beta_2'$ is a $\PITL$ theorem written in the vocabulary ${\bf L}^\nu\cup C^d_{\leq k}\cup C^p_{\leq k}$,
\begin{equation}\label{ruleapp}
\vdash_{\PITL_{{\bf L}^\nu\cup C^d_{\leq k}\cup C^p_{\leq k}}}(\ell=c\wedge\theta^\nu;\ell=\infty)\Rightarrow(\varphi^\nu\Rightarrow\bot),
\end{equation}
Now by an application of the rule $P_\leq$ to (\ref{ruleapp}), where the flexible symbols have no other superscript except $\nu$ as required by our restricted way of applying this $\PITL$-specific rule, we obtain
\[\vdash_{\PITL^-_{{\bf L}_{\leq k+1}'}}\ell=c\wedge\theta^\nu\wedge\ell<\infty\Rightarrow p(\varphi^\nu)\leq p(\bot)\]
which implies
\[\vdash_{\PITL^-_{{\bf L}_{\leq k+1}'}}\ell=c\wedge\theta^\nu\wedge\ell<\infty\Rightarrow p(\varphi^\nu)=0\]
by $P_\bot$ and, finally,
\[\vdash_{\PITL^-_{{\bf L}_{\leq k+1}'}}(\ell=c\wedge\theta^\nu\wedge\ell<\infty;\ell=\infty)\Rightarrow( p(\varphi^\nu)=0\wedge\ell=c;\ell=\infty)\]
by an application of the $\ITL$ proof rule $\Mono$.
Since $c<\infty,(\ell=c\wedge\theta^\nu;\ell=\infty)\in\Cn_{{\bf L}_{\leq k+1}'}(\Delta_k\cup\Xi_{\overline{\nu'}})$, this implies $(p(\varphi^\nu)=0\wedge\ell=c;\ell=\infty)\in\Cn_{{\bf L}_{\leq k+1}'}(\Delta_k\cup\Xi_{\overline{\nu'}})$. Hence 
$\Delta_k\cup\Xi_{\overline{\nu'}}$ is just as inconsistent as $\Delta_k\cup\Xi$, because the reason for all the formulas with flexible symbols superscripted by $\nu'=\langle\nu,c,\varphi\rangle$ to be in the finite subset $\Xi$ of $\Delta_{k+1}'$ is  
$(p(\varphi^\nu)\not=0\wedge\ell=c;\ell=\infty)\in\Delta_k$.
We can continue by showing that taking away the formulas of the form (\ref{addedformulas}) for some other superscript $\nu''\in S_{k+1}\setminus S_{\leq k}$ leads to a subset $(\Xi_{\overline{\nu'}})_{\overline{\nu''}}$ of $\Xi_{\overline{\nu'}}$ such that $\Delta_k\cup(\Xi_{\overline{\nu'}})_{\overline{\nu''}}$ is still inconsistent, etc., until there are no more symbols with superscripts from $S_{k+1}\setminus S_{\leq k}$ in the remaining subset of $\Xi$, which then will be a subset of $\Delta_k$. 
This is the sought contradiction, because we assume that $\Delta_k$ is consistent.
\end{proof}

For an arbitrary $k<\omega$, if $\Delta_{k+1}'$ is consistent, then $\Delta_{k+1}$ is defined as some maximal consistent set which contains $\Delta_{k+1}'$ and has witnesses in $C^d_{k+1}\cup C^p_{k+1}$. Its existence follows from Theorem \ref{lindenbaum} again. Then Lemma \ref{greatlemma} implies that all the sets in the sequence (\ref{deltasequence}) are consistent. Furthermore, obviously $\Delta$ is a maximal consistent set in ${\bf L}_D$ with respect to $\vdash_{\PITL^-}$ and has witnesses in $C^d\cup C^p$. The construction of $\Delta$ is complete.

\subsection{The $\PITL$ model $M$}
\label{modelm}

Since $\Delta$ is a maximal consistent set of $\PITL$ formulas written in ${\bf L}_D$ with witnesses in $C^d\cup C^p$, $\tr(\Delta)$ is maximal consistent set of $\ITL$ formulas written in $({\bf L}_D)_\ITL$ with witnesses in $C^d\cup C^p$ too. We use this to construct the model $M$ at two steps, the first being the construction of a canonical $\ITL$ model $M_\ITL$ which satisfies $\tr(\Delta)$ and the second being the construction of $M$ itself. This way we avoid the repetition of the non-$\PITL$-specific steps in the construction of $M$ which are as in \cite{WX04}.

\subsubsection{The $\ITL$ counterpart of $M$}

Let
\[c_1\equiv c_2\mbox{ iff }c_1=c_2\in\Delta\]
for constants $c_1,c_2\in C^d$ and $c_1,c_2\in C^p$.
Clearly, $\equiv$ is an equivalence relation on the constants from $C^d\cup C^p$. Let $[c]$ denote the $\equiv$-equivalence class which contains $c$ for each $c\in C^d\cup C^p$. Let 
\[T=\{[c]:c\in C^d\},\ D=T,\mbox{ and }U=\{[c]:c\in C^d\}.\]
Let  
\[[c']\leq[c'']\mbox{ iff }c'\leq c''\in\Delta\]
for $c',c''\in C^d$.
Clearly, $\leq$ is a linear ordering on $T$. Let $c_\infty$ be a witness in $C^d$ for the formula $\exists x(x=\infty)$ in $\Delta$. Then clearly $\langle T,\leq,[c_\infty]\rangle$ is a time domain.

Given $[[c'],[c'']]\in\tilde{\bf I}(T)$, we denote the set of formulas written in ${\bf L}_D$
\[\{\varphi:((\ell=c';\varphi)\wedge\ell=c'';\top)\vee(c''=\infty\wedge(\ell=c';\varphi))\in\Delta\}\]
by $\Delta_{[[c'],[c'']]}$. To understand the definition of $\Delta_{[[c'],[c'']]}$, recall our choice to start from a set $\Gamma$ such that $\ell=\infty\in\Gamma$ and, consequently, $\ell=\infty\in\Delta$. Let $c_0\in C^d$ be a witness for $\exists x(x=0)$ in $\Delta$ and $\sigma_0=[[c_0],[c_\infty]]$ for the rest of the section. Then obviously $\Delta_{\sigma_0}=\Delta$ and 
\begin{equation}\label{trdelta}
\varphi\in\Delta_{[[c'],[c'']]}\mbox{ iff }(\ell=c';\varphi)\in\Delta_{[[c_0],[c'']]}
\end{equation}
for all $\varphi\in{\bf L}_D$.

We define the mapping $I_\ITL$ of $({\bf L}_D)_\ITL$ by the clauses:

$I_\ITL(x),I_\ITL(d)\in A$ for individual variables $x$ and constants $d$ where $A=D$ for $x$ and $d$ of the duration sort and $A=U$ otherwise, and
\[I_\ITL(x)=\{c\in C^d\cup C^p:c=x\in\tr(\Delta)\},\ I_\ITL(d)=\{c\in C^d\cup C^p:c=d\in\tr(\Delta)\}.\]

$I_\ITL(f):A_1\times\ldots\times A_{\# f}\rightarrow A_{\# f+1}$ rigid function symbols $f$ where $A_1,\ldots,A_{\# f+1}$ are either $D$ or $U$, depending on the sort of the respective arguments of $f$ and the sort of its value, and 
\[I_\ITL(f)([c_1],\ldots,[c_{\# f}])=\{c\in C^d\cup C^p:c=f(c_1,\ldots,c_{\# f})\in\tr(\Delta)\}.\]

$I_\ITL(R):A_1\times\ldots\times A_{\# R}\rightarrow\{0,1\}$ for rigid relation symbols $R$ where $A_1,\ldots,A_{\# R}$ are as for function symbols, and 
\[I_\ITL(R)([c_1],\ldots,[c_{\# R}])=1\mbox{ iff }R(c_1,\ldots,c_n)\in\tr(\Delta).\]

$I_\ITL(d):\tilde{\bf I}(T)\rightarrow A$, $I_\ITL(f):\tilde{\bf I}(T)\times A_1\times\ldots\times A_{\# f}\rightarrow A_{\# f+1}$ and\\ $I_\ITL(R):A_1\times\ldots\times A_{\# R}\rightarrow\{0,1\}$ for flexible $d$, $f$ and $R$, respectively, where the $A$s are as for rigid symbols.
\[I_\ITL(d)(\sigma)=\{c\in C^d\cup C^p:c=d\in\tr(\Delta_{\sigma})\}.\]
Similarly, 

\[I_\ITL(f)(\sigma,[c_1],\ldots,[c_{\# f}])=\{c\in C^d\cup C^p:c=f(c_1,\ldots,c_{\# f})\in\tr(\Delta_{\sigma})\}.\]
Finally, $I_\ITL(R)(\sigma,[c_1],\ldots,[c_{\# R}])=1$ iff $R(c_1,\ldots,c_{\# R})\in\tr(\Delta_{\sigma})$.

\noindent
A lengthy but otherwise straighforward argument, which is standard for canonical models, shows that the above definitions are correct, $\langle D,I_\ITL(+),I_\ITL(0),I_\ITL(\infty)\rangle$ is a duration
domain, $\langle U,I_\ITL(+),I_\ITL(0),I_\ITL(1)\rangle$ is a probability
domain and $I_\ITL(\ell)$ is a measure function from $\tilde{\bf I}(T)$ to $D$,
\[F=\langle \langle T,\leq,I_\ITL(\infty)\rangle,\langle
D,I_\ITL(+),I_\ITL(0),I_\ITL(\infty)\rangle,\langle
U,I_\ITL(+),I_\ITL(0),I_\ITL(1)\rangle,I(\ell)\rangle\]
is a two-sorted frame for $\ITL$ with infinite intervals and $I$
is an $\ITL$ interpretation of $({\bf L}_D)_\ITL$ into $F$, which means
that $M_\ITL=\langle F,I_\ITL\rangle$ is a two-sorted $\ITL$ model for $({\bf L}_D)_\ITL$.
The standard truth lemma holds for $M_\ITL$, which is a canonical model:

\begin{lem}[Truth Lemma for $M_\ITL$]
Let $\sigma\in\tilde{\bf I}(T)$. Then
\[(I_\ITL)_\sigma(t)=\{c\in C^d\cup C^p:t=c\in\tr(\Delta_\sigma)\}\mbox{ and }M_\ITL,\sigma\models\varphi\mbox{ iff }\varphi\in\tr(\Delta_\sigma)\]
for every term $t$ and every formula $\varphi$ written in the vocabulary $({\bf L}_D)_\ITL$.
\end{lem}

\subsubsection{The model $M$}

Our next step is to define the $\PITL$ model $M=\langle F,{\bf W},I,P\rangle$ itself. The vocabulary of $M$ is ${\bf L}\cup C^d\cup C^p$ and its frame is $F$. Let $\Pi$ denote the set of the functions $\pi:V\rightarrow D\cup U$ where $V$ is a finite set of individual variables in ${\bf L}$ and $\pi(x)$ is in the domain which corresponds to the sort of $x$ for each $x\in V$. We define ${\bf W}$ as the set $S\times\Pi$.
Given $\nu\in S$, we define the interpretation $I_\nu$ by the equalities
\[I_\nu(s)=I_\ITL(s)\]
for rigid $s\in{\bf L}\cup C^d\cup C^p$, including the individual variables,
\[I_\nu(\ell)=m\mbox{ and }I_\nu(d)=I_\ITL(d^\nu)\]
for flexible constants $d\in{\bf L}\setminus\{\ell\}$ and
\[I_\nu(s)(\sigma,a_1,\ldots,a_{\# s})=I_\ITL(s^\nu)(\sigma,a_1,\ldots,a_{\# s})\]
for other flexible $s\in{\bf L}$. Now ${\bf W}$ consists of all the variants of the $I_\nu$ for all $\nu\in S$.

Given $w=\langle\nu,\pi\rangle$ such that $\dom\pi=\{x_1,\ldots,x_n\}$, we put \[I^w=(I_\nu)_{x_1,\ ,\ \ldots\ ,\ x_n}^{\pi(x_1),\ldots,\pi(x_n)}.\]
Some auxiliary notation is needed for the definition of $P^w$.

Let $\varphi$ be a formula written in ${\bf L}\cup C^d\cup C^p$, $FV(\varphi)=\emptyset$, $\nu\in S$ and $[[c'],[c'']]\in\tilde{I}(T)$. Then we denote the set
\[\{\nu'\in S:\varphi^{\nu'}\in\Delta_{[[c'],[c_\infty]]},
(\Box\forall(\chi^\nu\Leftrightarrow\chi^{\nu'})\wedge\ell=[c''];\top)\in\Delta_{[[c'],[c_\infty]]}\mbox{ for all }\chi\mbox{ in }{\bf L}\cup C^d\cup C^p\}\]
by $S_{\varphi,\nu,[[c'],[c'']]}$. We use $S_{\varphi,\nu,[[c'],[c'']]}$ to define a syntactical conterpart $\synsem{.}$ to $\sem{.}$ in our model under construction. If $\psi$ is a formula written in ${\bf L}\cup C^d\cup C^p$, $FV(\psi)=\{x_1,\ldots,x_n\}$ and $c_i\in I^{\langle\nu,\pi\rangle}(x_i)$, $i=1,\ldots,n$, then we put
\begin{equation}\label{defset}
\synsem{\varphi}_{\langle\nu,\pi\rangle,[[c'],[c'']]}=\{\langle \nu',\pi'\rangle\in{\bf W}:\nu'\in S_{[c_1/x_1,\ldots,c_n/x_n]\varphi,\nu,[[c'],[c'']]},\pi'\in\Pi\}.
\end{equation}
Clearly, the set on the right of $=$ in (\ref{defset}) does not depend on the precise choice of $c_i\in I^{\langle\nu,\pi\rangle}(x_i)$, $i=1,\ldots,n$.
The truth lemma about $M$ which is proved below entails that 
\begin{equation}\label{synsemsem}
\synsem{\varphi}_{w,[[c'],[c'']]}=\sem{\varphi}_{M,w,[[c'],[c'']]}.
\end{equation}
Note that 
\begin{equation}\label{extpast}
\synsem{\varphi}_{w,[[c'],[c'']]}=\synsem{(\ell=c';\varphi)}_{w,[[c_0],[c'']]}
\end{equation}
follows from (\ref{trdelta}) and therefore the rest of the construction steps involve mostly intervals $\sigma\in\tilde{I}(T)$ such that $\min\sigma=[c_0]$. Given $w\in{\bf W}$, $w=\langle\nu,\pi\rangle$, a formula $\varphi$ written in ${\bf L}\cup C^d\cup C^p$ whose free variables are $x_1,\ldots,x_n$, $\nu\in S$, $c_i\in I^w(x_i)$, $i=1,\ldots,n$, and $[c'']\in T$ we define $P^w$ on the subsets of ${\bf W}$ of the form (\ref{defset}) by the equality

\[P^w([c''],\synsem{\varphi}_{w,[[c_0],[c'']]})=\{c\in C^p:p([c_1/x_1,\ldots,c_n/x_n]\varphi^\nu)=c\in\Delta_{[[c_0],[c'']]}\}.\] 
For this definition to be correct, we need to have 
\[p([c_1/x_1,\ldots,c_n/x_n]\varphi^\nu)=c\in\Delta_{[[c_0],[c'']]}\mbox{ iff }p([c_1/x_1,\ldots,c_n/x_n]\psi^\nu)=c\in\Delta_{[[c_0],[c'']]}\]
for formulas $\varphi$ and $\psi$ such that 
\begin{equation}\label{synsemeq}
\synsem{\varphi}_{w,[[c_0],[c'']]}=\synsem{\psi}_{w,[[c_0],[c'']]},
\end{equation}
and $c_i\in I^w(x_i)$, $i=1,\ldots,n$, where $\{x_1,\ldots,x_n\}=FV(\varphi)\cup FV(\psi)$. To prove it, assume that \[p([c_1/x_1,\ldots,c_n/x_n]\varphi^\nu)<p([c_1/x_1,\ldots,c_n/x_n]\psi^\nu)\in\Delta_{[[c_0],[c'']]}\]
for the sake of contradiction. Then \[p([c_1/x_1,\ldots,c_n/x_n](\psi^\nu\wedge\neg\varphi^\nu))\not=0\in\Delta_{[[c_0],[c'']]}\]
by $\PITL 3$ from Section \ref{usefultheorems}. If $c''<\infty\in\Delta$, then this implies that \[\langle\langle\nu,c'',\psi\wedge\neg\varphi\rangle,\pi'\rangle\in\synsem{\psi}_{w,[[c_0],[c'']]}\setminus\synsem{\varphi}_{w,[[c_0],[c'']]}\]
where $\dom\pi'=FV(\varphi)\cup FV(\psi)$ and $\pi'(x_i)=I^w(x_i)$. $i=1,\ldots,n$, which contradicts (\ref{synsemeq}). If $c''=\infty\in\Delta$, then the appropriate instances of $P_\infty$ and $\PITL 2$ from Section \ref{usefultheorems} imply that 
\[p([c_1/x_1,\ldots,c_n/x_n](\psi^\nu\wedge\neg\varphi^\nu))=1\in\Delta_{[[c_0],[c'']]}\]
and, consequently, \[[c_1/x_1,\ldots,c_n/x_n](\psi^\nu\wedge\neg\varphi^\nu)\in\Delta_{[[c_0],[c'']]}.\]
This implies that $w$ itself is in $\synsem{\psi}_{w,[[c_0],[c'']]}\setminus\synsem{\varphi}_{w,[[c_0],[c'']]}$, which contradicts (\ref{synsemeq}) too.

The presence of all the instances of $P_\bot$, $P_\top$ and $P_+$ written in the vocabularies ${\bf L}^\nu\cup C^d\cup C^p$, $\nu\in S$, in $\Delta_{[[c_0],[c'']]}$ implies that $\lambda X.P^w([c''],X)$ is a finitely additive probability function on the boolean algebra
\[\langle\{\synsem{\psi}_{w,[[c_0],[c'']]}:\psi\in{\bf L}\},\cap,\cup,\emptyset,{\bf W}_{w,[c'']}\rangle\]
for every $w\in{\bf W}$ and every $[c'']\in T$. Note that this algebra contains the sets $\synsem{\psi}_{w,[[c'],[c'']]}$ for all $c'\in C^d$ such that $c'\leq c''\in\Delta$ because of (\ref{extpast}). Clearly, $M=\langle F,{\bf W},I,P\rangle$ is a $PITL$ model for the vocabulary ${\bf L}\cup C^d\cup C^p$. 

Obviously if $w=\langle\nu,\pi\rangle$ for some $\pi\in\Pi$ then $\{\langle\langle\nu,c,\varphi\rangle,\pi'\rangle:\pi'\in\Pi\}\subseteq\in{\bf W}_{w,[c]}$ for all $\nu\in S_{\leq k}$, $c\in C^d$ and all $\varphi$ written in ${\bf L}_{\leq k}$ such that $(p(\varphi^\nu)\not=0\wedge\ell=c;\top)\in\Delta$ and all $k<\omega$, because, according to the construction of $\Delta$, in this case
\[(\Box\forall(\chi^\nu\Leftrightarrow\chi^{\langle\nu,c,\varphi\rangle})\wedge\ell=c;\top)\in\Delta\]
for all formulas $\chi$ written in ${\bf L}\cup C^d\cup C^p$, and in particular for $\chi$ of the forms $d=x$, $f(x_1,\ldots,x_{\# f})=x_{\# f+1}$, $R(x_1,\ldots,x_{\# R})$ and $p(\psi)=x$ where $d$, $f$ and $R$ are flexible constants, function and relation symbols from ${\bf L}$, and $\psi$ is written in ${\bf L}\cup C^d\cup C^p$ respectively. 
Furthermore, if $I^w$ is a variant of $I^v$ and $P^w=P^v$ for some $w,v\in{\bf W}$, then ${\bf W}_{w,[c]}={\bf W}_{v,[c]}$ for all $[c]\in T$.

Here follows the truth lemma for $M$:

\begin{lem}[Truth Lemma for $M$]
\label{truthlemma}
Let $\sigma\in\tilde{\bf I}(T)$, $w\in{\bf W}$ and $w=\langle\nu,\pi\rangle$. If $t$ is a term written in ${\bf L}_D$, $FV(t)=\{x_1,\ldots,x_n\}$ and $c_1,\ldots,c_n\in C^d\cup C^p$ are such that $c_i\in I^w(x_i)$, $i=1,\ldots,n$, then
\[w_\sigma(t)=\{c\in C^d\cup C^p:[c_1/x_1,\ldots,c_n/x_n]t^\nu=c\in\Delta_\sigma\}.\]
If $\varphi$ is a formula written in ${\bf L}_D$, $FV(\varphi)=\{x_1,\ldots,x_n\}$ and $c_1,\ldots,c_n$ satisfy the same conditions as above, then
\[M,w,\sigma\models\varphi\mbox{ iff }[c_1/x_1,\ldots,c_n/x_n]\varphi^\nu\in\Delta_\sigma.\]
\end{lem}

We use the constants $c_1,\ldots,c_n$ in the formulation of the lemma, because we need it to apply to $w\in{\bf W}$ with variants to some interpretation of the form $I_\nu$, and not just to the interpretations $I_\nu$, $\nu\in S$, themselves.

\begin{proof}
The proof is by simultaneous induction on the length of terms and formulas. The clause of the lemma about formulas implies (\ref{synsemsem}).

The induction base and the steps for formulas and for terms built using constants, variables and function symbols are as in (non-probabilistic) $\ITL$ and we omit them. We only do the case of probabilistic terms $p(\psi)$. According to our definition, $FV(p(\psi))=FV(\psi)$. Let $x_1,\ldots,x_n$ and $c_1,\ldots,c_n$ be as in the lemma and $\sigma=[[c'],[c'']]$. 
Since 
\[
\begin{array}{lll}w_{[[c'],[c'']]}(p(\psi)) &=&P^w([c''],\sem{\psi}_{M,w,{[[c'],[c'']]}})\\ &=&P^w([c''],\sem{(\ell=c';\psi)}_{M,w,[[c_0],[c'']]})\\ &=&w_{[[c_0],[c'']]}(p((\ell=c';\psi)))\end{array}\]
and
\[[c_1/x_1,\ldots,c_n/x_n]p(\psi^\nu)=c\in\Delta_{[[c'],[c'']]}\mbox{ iff }[c_1/x_1,\ldots,c_n/x_n]p((\ell=c';\psi^\nu))=c\in\Delta_{[[c_0],[c'']]}\]
because of the instances $(\ell=c';p(\psi)=d)\Rightarrow p((\ell=c';\psi))=d$ of $P_;$, which are in $\Delta_{[c_0],[c'']}$ for all $d\in C^p$, 
it is sufficient to prove 
\begin{equation}\label{reducedp}
w_{[[c_0],[c'']]}(p((\ell=c';\psi)))=\{c\in C^d\cup C^p:p((\ell=c';[c_1/x_1,\ldots,c_n/x_n]\psi^\nu))=c\in\Delta_{[[c_0],[c'']]}\}.
\end{equation}
By the induction hypothesis, the lemma holds for $\psi$ and therefore
\[\synsem{\psi}_{w,[[c'],[c'']]}=\sem{\psi}_{M,w,[[c'],[c'']]},\]
which implies
\[\synsem{(\ell=c';\psi)}_{w,[[c_0],[c'']]}=\sem{(\ell=c';\psi)}_{M,w,[[c_0],[c'']]}\]
by (\ref{extpast}) and the definition of $\sem{.}_{M,w,[.,[c'']]}$. Now (\ref{reducedp}) follows from the definition of $P^w$.
\end{proof}

We conclude the presentation of $M$ with the observation that $S$ and the domains in $F$ are countably-infinite and therefore every interpretation in ${\bf W}$ has only countably many variants, which entails that ${\bf W}$ is a countably-infinite set. 

\subsection{The completeness theorem}

Now it is easy to prove the strong completeness theorem for our proof system for $\PITL$.

\begin{thm}

Let ${\bf L}$ be a $\PITL$ vocabulary and $\Gamma$ be a set of formulas written in ${\bf L}$ which is consistent with the proof system from Section \ref{proofsystem}. Then there exists a model $M_\Gamma=\langle F_\Gamma,{\bf W}_\Gamma,I_\Gamma,P_\Gamma\rangle$ for ${\bf L}$ and an $w_0\in{\bf W}_\Gamma$ and a time interval $\sigma_0$ in it such that \begin{equation}\label{compl}
M_\Gamma,w_0,\sigma_0\models\varphi\mbox{ for all }\varphi\in\Gamma.
\end{equation}
\end{thm}

\begin{proof}
If $\Gamma$ is consistent with the formula $\ell=\infty$, then we can take the model $M=\langle F,{\bf W},I,P\rangle$ constructed in Section \ref{modelm} for $\Gamma\cup\{\ell=\infty\}$. Otherwise $\Gamma$ is consistent with the formula $\ell=c\wedge c<\infty$ for some rigid constant $c\not\in{\bf L}$ and we can take $M$ from Section \ref{modelm} for the set (\ref{infiniteell}). In both cases $M_\Gamma$ can be chosen to be $\langle F,{\bf W},\lambda w.(I^w|_{\bf L}),P\rangle$ where $I^w|_{\bf L}$ stands for the restriction of $I_w$ to the initially given vocabulary ${\bf L}$, and $w_0$ can be chosen to be $\langle\langle\rangle,\emptyset\rangle$ where $\langle\rangle$ is the only element of $S_0$ and $\emptyset$ denotes the empty function $\emptyset\rightarrow C^d\cup C^p$. In the first case the interval $\sigma_0$ can be chosen to be the entire time domain $T$ of $F$. In the second case $\sigma_0$ can be chosen to be $[\min T,I^{w_0}(c)]$ where $c$ is the constant introduced above. The equivalence now follows from the definition of $\Delta$ and Lemma \ref{truthlemma}.
\end{proof}

\section{Axioms for global probability in $\PITL$ models}
\label{globalprobab}

We call the models for $\PITL$ introduced in Definition \ref{pitlmodels} {\em general}, because the probability functions $\lambda X.P^w(\tau,X)$ in them can be arbitrary, whereas it is natural to require these functions to satisfy certain constraints. Applications typically lead to models in which all the probability functions originate from a {\em global} probability function on the entire ${\bf W}$ such as the automata-based models of $\PDC$. Consider models $M=\langle F,{\bf W},I,P\rangle$ with frames $F=\langle\langle T,\leq,\infty\rangle,\langle D,+,0,\infty\rangle,\langle U,+,0,1\rangle,m\rangle$ whose time domain has a least element $\tau_0=\min T$ and a distinguished $w_0\in{\bf W}$ such that ${\bf W}_{w_0,\tau_0}={\bf W}$. Then $\lambda X.P^{w_0}(\tau_0,X)$ can be regarded as the global probability function and, given an arbitrary $w\in{\bf W}$ and $\tau\in T$, the probability function $\lambda X.P^w(\tau,X)$ should represent {\em conditional} probability on sets of interpretations, the condition being $\tau$-equivalence with $w$. Hence we should have
\begin{equation}
\label{condprobab}
P^{w_0}(\tau_0,{\bf W}_{w,\tau}).P^w(\tau,A)=P^{w_0}(\tau_0,{\bf W}_{w,\tau}\cap A)
\end{equation}
with respect to an appropriately defined operation of multiplication $.$ on the probability domain for all $A\subseteq{\bf W}$ such that the above equality is defined. This equality is usually insufficient to determine $\lambda X.P^w(\tau,X)$, because, e.g., it is possible that $P^{w_0}(\tau,{\bf W}_{w,\tau})=0$. A more general constraint of this form can be formulated as follows. Let $M$, $w$ and $A\subseteq{\bf W}$ be as above, $\tau,\tau'\in T$ and $\tau\leq\tau'$. Then
\begin{equation}\label{fullprobab}
P^{w_0}(\tau,A)=\int\limits_{w\in{\bf W}_{w_0,\tau}}P^w(\tau',A)d(\lambda X.P^{w_0}(\tau,X)).
\end{equation}
The integral above is not guaranteed to exist for an arbitrary probability domain, because its definition involves least upper bounds and greatest lower bounds of sets of approximating sums, which may be unavailable if there are Dedekind gaps, which is the case if, e.g., the probability domain is based the non-negative rational numbers. Dedekind-completeness is not a first-order property and therefore our proof system for $\PITL$ cannot be extended to one that is complete with respect to Dedekind-complete domains by finitary means. In this section we propose axioms which enforce the best possible approximation of (\ref{fullprobab}) permitted by the probability domain.

In the rest of the section we consider $\PITL$ models $\langle F,{\bf W},I,P\rangle$ with the probability domains of their frames $F$ extended to have multiplication. Given \\ $F=\langle\langle T,\leq,\infty\rangle,\langle D,+,0,\infty\rangle,\langle U,+,.,0,1\rangle,m\rangle$, we assume that the new operation satisfies, e.g., the following axioms:

\qquad

\begin{tabular}{ll}
$(U8)$ & $(x.y).z=x.(y.z)$\\
$(U9)$ & $ x.y=y.x$\\
$(U10)$ & $(x+y).z=x.z+y.z$\\
$(U11)$ & $x.1=x$\\
$(U12)$ & $x.y=x.z\Rightarrow x=0\vee y=z$\\
$(U13)$ & $x=0\vee\exists y(x.y=z)$\\
\end{tabular}

\qquad

\noindent
Together with $(U1)$-$(U7)$, these axioms are sufficient to extend a probability domain to a field by introducing negative elements and division in the customary way.

We adopt a definition for the integral in (\ref{fullprobab}) which is based on Darboux-Lebesgue sums as known from the theory of integration of real-valued functions. Let the measurable sets $B_0,\ldots,B_n$ form a partition of ${\bf W}_{w_0,\tau}$ and let $P^w(\tau',A)\in[\xi_i,\eta_i]$ for all $w\in B_i$, $i=0,\ldots,n$. Then the sums 
\begin{equation}\label{approxsums}
\sum\limits_{i=0}^n \xi_iP^{w_0}(\tau,B_i)\mbox{ and }\sum\limits_{i=0}^n \eta_iP^{w_0}(\tau,B_i)
\end{equation}
are a lower and an upper approximation for the integral from (\ref{fullprobab}), respectively. The integral is defined if both the least upper bound of the lower approximations and the greatest lower bound of the upper approximations of the above forms taken for all partitions $B_0,\ldots,B_n$ of ${\bf W}_{w,\tau}$ into measurable subsets and all appropriate boundary probabilities $\xi_i$, $\eta_i$, $i=0,\ldots,n$, exist and are equal.

The sets $A$ for which $P^{w_0}(\tau,A)$ and $P^w(\tau',A)$, $w\in{\bf W}_{w_0,\tau}$ need to be defined have the forms $\sem{\varphi}_{M,w_0,[\tau'',\tau]}$ and $\sem{\varphi}_{M,w,[\tau'',\tau']}= \sem{\varphi}_{M,w_0,[\tau'',\tau]}\cap{\bf W}_{w,\tau'}$, respectively, where $\varphi$ is a formula in the vocabulary of $M$ and $\tau''\leq\tau$. Hence (\ref{fullprobab}) can be written as
\begin{equation}\label{fullprobab2}
P^{w_0}(\tau,\sem{\varphi}_{M,w_0,[\tau'',\tau]})=\int\limits_{w\in{\bf W}_{w_0,\tau}}P(\tau',\sem{\varphi}_{M,w,[\tau'',\tau']})d(\lambda X.P^{w_0}(\tau,X)).
\end{equation}
Our axioms for (\ref{fullprobab2}) exploit the observation that the sets which are available for the construction of partitions $B_0,\ldots,B_n$ have such forms too. Here they are:
\[\begin{tabular}{ll}
$\!(\overline{P})\!\!$ & $\ell\leq y\wedge p((\ell=y\wedge\theta\wedge p(\varphi)>x;\top))=0\Rightarrow p((\theta\wedge\ell=y;\top)\wedge\varphi)\leq x.p((\theta\wedge\ell=y;\top))$\\
$\!(\underline{P})\!\!$ & $\ell\leq y\wedge p((\ell=y\wedge\theta\wedge p(\varphi)\leq x;\top))=0\Rightarrow p((\theta\wedge\ell=y;\top)\wedge\varphi)\geq x.p((\theta\wedge\ell=y;\top))$\\
\end{tabular}
\]
Let us show that these axioms enforce the possible approximations of (\ref{fullprobab2}). Assume that $\overline{P}$ and $\underline{P}$ are part of our proof system. Let $\varphi$ be a $\PITL$ formula, $y$ be an individual variable of the duration sort and $x_0,\ldots,x_n$ be $n+1$ individual variables of the probability sort. Let
\[\theta_0\rightleftharpoons p(\varphi)\leq x_0,\ \theta_i\rightleftharpoons x_{i-1}< p(\varphi)\wedge p(\varphi)\leq x_i,\ i=1,\ldots,n.\]
Now consider the instances 
\[\ell\leq y\wedge p((\ell=y\wedge\theta_i\wedge p(\varphi)>x_i;\top))=0\Rightarrow p((\theta_i\wedge\ell=y;\top)\wedge\varphi)\leq x_i.p((\theta_i\wedge\ell=y;\top))\]
\[\ell\leq y\wedge p((\ell=y\wedge\theta_i\wedge p(\varphi)\leq x_{i-1};\top))=0\Rightarrow p((\theta_i\wedge\ell=y;\top)\wedge\varphi)\geq x_{i-1}.p((\theta_i\wedge\ell=y;\top))\]
of $\overline{P}$ and $\underline{P}$ for $i=1,\ldots,n$ and the instance
\[\ell\leq y\wedge p((\ell=y\wedge\theta_0\wedge p(\varphi)>x_0;\top))=0\Rightarrow p((\theta_0\wedge\ell=y;\top)\wedge\varphi)\leq x_0.p((\theta_0\wedge\ell=y;\top)) \]
of $\overline{P}$. Since
\[\vdash_{\PITL}\theta_i\wedge p(\varphi)>x_i\Rightarrow\bot\mbox{ and }\vdash_{\PITL}\theta_i\wedge p(\varphi)\leq x_{i-1}\Rightarrow\bot,\]
we have 

\[\vdash_{\PITL}p((\ell=y\wedge\theta_i\wedge p(\varphi)>x_i;\top))=0,\ p((\ell=y\wedge\theta_i\wedge p(\varphi)<x_{i-1};\top))=0\]
by $\PITL 1$ and $P_\bot$.
Hence the considered instances of $\overline{P}$ and $\underline{P}$ entail
\begin{equation}\label{papprox1}
\vdash_{\PITL}\ell\leq y\Rightarrow x_{i-1}.p((\theta_i\wedge\ell=y;\top))\leq p((\theta_i\wedge\ell=y;\top)\wedge\varphi)
\end{equation}
for $i=1,\ldots,n$ and 
\begin{equation}\label{papprox2}
\vdash_{\PITL}\ell\leq y\Rightarrow p((\theta_i\wedge\ell=y;\top)\wedge\varphi)\leq x_i.p((\theta_i\wedge\ell=y;\top))
\end{equation}
for $i=0,\ldots,n$. Let $\chi$ denote the rigid formula
\[y<\infty\wedge x_0=0\wedge x_n=1\wedge\bigwedge\limits_{i=1}^{n}x_{i-1}\leq x_i.\]
Then a purely $\ITL$ deduction shows that
\[\vdash_{\PITL}\chi\Rightarrow\left( \varphi\Leftrightarrow\bigvee\limits_{i=0}^n((\theta_i\wedge\ell=y;\top)\wedge\varphi)\right)\]
and
\[\vdash_{\PITL}\chi\Rightarrow\neg(((\theta_i\wedge\ell=y;\top)\wedge\varphi)\wedge((\theta_j\wedge\ell=y;\top)\wedge\varphi))\]
for $i\not=j$, $i,j=0,\ldots,n$. Hence, using the axioms for arithmetics of probabilities and $\PITL 4$, we can derive
\[\vdash_{\PITL}\chi\Rightarrow p(\varphi)=\sum\limits_{i=0}^n p((\theta_i\wedge\ell=y;\top)\wedge\varphi).\]
Now (\ref{papprox1}) and (\ref{papprox2}) imply
\begin{equation}\label{papprox3}
\vdash_{\PITL}\chi\Rightarrow\sum\limits_{i=1}^n x_{i-1}.p((\theta_i\wedge\ell=y;\top))\leq p(\varphi)\wedge p(\varphi)\leq\sum\limits_{i=0}^n x_i.p((\theta_i\wedge\ell=y;\top)).
\end{equation}
Recall the model $M$ and its distinguished $w_0\in{\bf W}$ and time point $\tau_0$. Let $\tau,\tau'\in T$ and $\tau\leq\tau'$. Let $I^{w_0}(y)=m([\tau_0,\tau'])$. Then the satisfaction of (\ref{papprox3}) at $w_0,[\tau_0,\tau]$ in $M$ means that if $A=\sem{\varphi}_{M,w_0,\tau}$ and $B_i=\sem{\theta_i}_{M,w_0,\tau}$, $i=0,\ldots,n$, then $P^{w_0}(\tau,A)$ is bounded by the sums (\ref{approxsums}) where $\xi_0=0$, $\eta_0=I^{w_0}(x_0)$ and $\xi_i=I^{w_0}(x_{i-1})$ and $\eta_i=I^{w_0}(x_i)$ for $i=1,\ldots,n$. Assume that $z$ is a variable of the probability sort and $M$ satisfies the rigid formula
\[\bigwedge\limits_{i=1}^{n}x_i\leq x_{i-1}+z\]
at $w_0$ as well. Then, since $\sum\limits_{i=0}^nP^{w_0}(\tau,B_i)=1$, the lower and upper approximations (\ref{approxsums}) differ by no more than $I^{w_0}(z)$. Now it is clear that the validity of $\overline{P}$ and $\underline{P}$ in $M$ entails that (\ref{fullprobab2}) holds approximately with precision which is smaller than any probability $\delta\in U$ such that $\underbrace{\delta+\ldots+\delta}_{n\ \mbox{\scriptsize times}}\geq 1$ for some $n<\omega$. Hence, if $\langle U,+,.,0,1\rangle$ has no ``infinitely small'' elements, then the integral from (\ref{fullprobab2}) is defined and (\ref{fullprobab2}) holds. If there are such elements, then the difference between the least upper bound and the greatest lower bound of the sums (\ref{approxsums}), respectively, is ``infinitely small''. 

Obviously the condition ${\bf W}_{w_0,\tau_0}={\bf W}$ is relevant just to the {\em scope} of the (approximate) validity of (\ref{fullprobab}). If all instances of $\overline{P}$ and $\underline{P}$ hold everywhere in a $\PITL$ model, then so do the approximations of (\ref{fullprobab}).

\section{Probabilistic real-time $\DC$ with infinite intervals}

In this section we introduce an enhanced system of real-time probabilistic $\DC$ which enables the handling of infinite intervals and has a syntactically simpler and more expressive probability operator instead of the original $\mu(.)(.)$. The new system is obtained as the extension of $\PITL$ by state expressions and duration terms. It properly subsumes the original probabilistic real-time $\DC$ from \cite{DZ99} in a straightforward way. The relative completeness result about probabilistic $\DC$ in this paper is about this enhanced system and we use the acronym $\PDC$ for it in the rest of the paper.

\subsection{Language}

$\PDC$ vocabularies are just $\PITL$ vocabularies extended by state variables, which are used to construct state expressions and duration terms just like in (non-probabilistic) $\DC$ (see Section \ref{dcprelim} of the Preliminaries).  

\subsection{Models and satisfaction}
\label{pdcmodels}

$\PDC$ models are $\PITL$ models which are based on the real-time and -probability frame for two-sorted $\ITL$ with infinite intervals 
\[F_{\bf R}=\langle\langle\overline{\bf R},\leq,\infty\rangle,\langle\overline{\bf R}_+,+,0,\infty\rangle,\langle\overline{\bf R}_+,+,.,0,1\rangle,\lambda\sigma.\max\sigma-\min\sigma\rangle,\]
the only difference being that the interpretations $I^w$, $w\in{\bf W}$ are supposed to map the state variables from the respective vocabularies to $\{0,1\}$-valued functions of time with the finite variability property. We assume that multiplication is available for probabilities. The definition of the values of duration terms and the definition of the satisfaction relation are just like in $\DC$ and $\PITL$, respectively.

\subsection{Describing probabilistic real-time automata and expressing $\mu(.)(.)$}
\label{modellingautomata}

The probabilistic automata from the semantics of $\PDC$ originally introduced in \cite{DZ99} can be described in the system of $\PDC$ proposed in this paper. The original probability operator $\mu(.)(.)$ can be expressed using $p(.)$ as follows.

Let ${\bf A}$ be an automaton of the form (\ref{fpta}) from Definition \ref{pautomata}. The $\DC$ vocabulary which corresponds to ${\bf A}$ consists the states of ${\bf A}$ as state variables and the $\PITL$ vocabulary for ${\bf A}$ introduced the example from Section \ref{pitlsemantics}, which includes the transitions of ${\bf A}$ as temporal propositional letters ($0$-ary flexible predicate symbols), the rigid constants $q_a$ and the rigid unary function symbols $P_a$ to denote $\lambda\tau.\int\limits_0^\tau p_a(t)dt$ for each transition $a$, respectively. Let $M=\langle F_{\bf R},{\bf W},I,P\rangle$ be a $\PDC$ model for this vocabulary in the sense of Section \ref{pdcmodels} with ${\bf W}$ being the set of all the behaviours of ${\bf A}$ and $\lambda X.P^w(\tau,X)$ being the conditional probability for a behaviour of ${\bf A}$ to be described by an interpretation in the set $X\subseteq{\bf W}_{w,\tau}$, given that $w\in{\bf W}$ describes this behaviour within the interval $[0,\tau]$, like in the example from Section \ref{pitlsemantics}. Then $M$ validates the axioms
\[\Box\neg(\pred{\neg a^-};\pred{a^-}\wedge\neg a;\pred{a^+}),
\ \neg(\pred{a^-}\wedge\neg a;\pred{a^+};\top)\]
and 
\[\Box(\neg(\pred{a^-};a)\wedge\neg(a\wedge\neg\pred{a^-})\wedge\neg(a;\pred{\neg a^+}))\]
for all transitions $a$ at all intervals $\sigma$ such that $\min\sigma=0$. These axioms force the interpretations of the temporal propositional letters $a$ to correspond to the respective transitions of ${\bf A}$, which are identified by observing their source states $a^-$ and destination states $a^+$, in the way proposed in the example from Section \ref{pitlsemantics}. Having this correspondence, the probabilistic behaviour of ${\bf A}$ can be described by formulas such as (\ref{exprobabformula}). If used together with the axioms $\overline{P}$ and $\underline{P}$ from Section \ref{globalprobab}, such formulas are sufficient to express the conditions on the probability functions $\lambda X.P^w(\tau,X)$ for $w\in{\bf W}$ which are encoded by the components $p_a$ and $q_a$ of the automaton ${\bf A}$. Furthermore, the value of $\mu(\varphi)(t)$ is equal to $w_{[0,0]}(p((\varphi\wedge\ell=t;\top)))$ for every $\DC$ formula $\varphi$ and every $w\in{\bf W}$.

Note that the probabilities expressed by terms of the form $p(\varphi)$ are determined by using the truth values of $\varphi$ at infinite intervals. That is why the probability for $\varphi$ to hold at a finite interval ending at some future time point is expressed by the term $p((\varphi;\top))$, in which $\top$ accounts of the infinite interval following that end point.

In our $\PDC$ axioms about probabilistic timed automata behaviour we refer to the probability $P_a(\tau)$ for transition $a$ to be over by time $\tau$ instead of the probability density $p_a(t)$ for $a$ to finish at time $t$, which was used in the original paper \cite{DZ99}. This is not a limitation, because, at least in the case of piece-wise continuous $p_a$, the relation $P_a(\tau)=\int\limits_0^\tau p_a(t)dt$ between $P_a$ and $p_a$ can be axiomatised much like (\ref{fullprobab}). On the contrary, there are practically interesting cases such as that of transitions with discrete or finite sets of possible durations in which $p_a$ cannot be defined whereas $P_a$ exists. 

\section{A proof system for $\PDC$}

The proof system for $\PDC$ that we propose consists of the $\DC$ axioms $\DC 1$-$\DC 6$, $T1$ and $T2$ from Section \ref{dcinfax}. We demonstrate the relative completeness of this proof system in Section \ref{relativecompleteness} below. Since completeness relative to validity in the class of the $\PITL$ models which are based on $F_{\bf R}$ means that all formulas which are valid at such $\PITL$ models are admitted as axioms, the $\PITL$ axioms from Section \ref{proofsystem} are no more relevant than any of these valid formulas from the formal point of view. 

\section{Relative completeness of the proof system for $\PDC$}
\label{relativecompleteness}

The proof of the completeness of the axioms $\DC 1$-$\DC 6$, $T1$ and $T2$ for $\PDC$ relative to validity in the class of the $F_{\bf R}$-based models of $\PITL$ follows closely the pattern of the original relative completeness proof for (non-probabilistic) $\DC$ from \cite{HZ92}. The variant of this proof about the system of $\DC$ based on the modalities of $\NL$ from \cite{RZ97} is very close to our setting. Therefore we include the proof details mostly for the sake of completeness. Below $\PITL^{\bf R}_{\bf L}$ stands for the set of the $\PITL$ formulas written in the vocabulary ${\bf L}$ which are valid in the class of all $F_{\bf R}$-based $\PITL$ models.

Let $\varphi$ be a $\PDC$ formula written in some vocabulary ${\bf L}$ and let ${\bf S}$ be the set of all the state expressions which can be written using only the state variables which occur in $\varphi$. Given a state expression $S\in{\bf S}$, we denote the set
\[\{S'\in{\bf S}:S'\mbox{ is propositionally equivalent to }S\}\]
by $[S]$. Since $\varphi$ contains a finite number of state variables, there are finitely many different equivalence classes $[S]$ for $S\in{\bf S}$. Let ${\bf L}'$ be the $\ITL$ vocabulary which consists of the symbols from ${\bf L}$, except the state variables, and the fresh flexible constants $\ell_{[S]}$, $S\in{\bf S}$. Since there are finitely many classes $[S]$, these flexible constants are finitely many too. If all the state expressions which occur in some $\PDC$ formula $\psi$ are from ${\bf S}$, we denote the result of substituting every duration term $\int S$ with the respective flexible constant $\ell_{[S]}$ in $\psi$ by $\psi'$. Note that $\psi'$ is a $\PITL$ formula with no $\PDC$-specific constructs left in it.

Now consider the set ${\bf H}$ of all the instances of $\DC 1$-$\DC 6$, $T1$ and $T2$ for state expressions from ${\bf S}$. Unless no state variables occur in $\varphi$, ${\bf H}$ is infinite. However, since there are finitely many equivalence classes $[S]$, the set
\[{\bf H}'=\{\alpha':\alpha\in{\bf H}\}\]
is finite. We define the sequence of formulas $\psi_k$, $k<\omega$ as follows:
\[\psi_0\rightleftharpoons\Box\bigwedge{\bf H}',\ \psi_{k+1}\rightleftharpoons\Box\bigwedge{\bf H}'\wedge p(\psi_k)=1\mbox{ for all }k<\omega.\]
The formula $\psi_k$ states that all the instances of the $\DC$ axioms hold with probability $1$ at interpretations which are accessible through probability terms of height at most $k$.

Now assume that $\varphi$ is consistent with our proof system for $\PDC$. Let $n=h(\varphi)$ where $h(\varphi)=0$ for $\varphi$ with no occurrence of probability terms, and 
$h(\varphi)=1+\max\{h(\psi):p(\psi)\mbox{ occurs in }\varphi\}$ for $\varphi$ with probability terms. Then the formula 
\[\psi\rightleftharpoons\ell=\infty\wedge(\varphi'\vee(\varphi';\ell=\infty))\wedge\psi_n\] 

is consistent with $\PITL^{\bf R}_{\bf L}$. This entails that there is a $\PITL$ model $M=\langle F_{\bf R},{\bf W},I,P\rangle$, $w_0\in{\bf W}$ and an interval $\sigma_0\in\tilde{\bf I}(\overline{\bf R})$ such that 
\[M,w_0,\sigma_0\models\psi.\]
Clearly $\sigma_0\in{\bf I}^\infin(\overline{\bf R})$. Following the example from \cite{HZ92}, we use $M$ in order to build a $\PDC$ model for ${\bf L}$ which satisfies $\varphi$.

We define the ascending sequence of subsets ${\bf N}_0\subseteq{\bf N}_1\subseteq\ldots\subseteq{\bf N}_n$ of ${\bf W}$ by the equalities
\[{\bf N}_0=\{w_0\}\mbox{ and }{\bf N}_k=\bigcup\limits_{w\in{\bf N}_{k-1}}\{v\in{\bf W}_{w,\min\sigma_0}:M,v,\sigma_0\models\psi_{n-k}\}\mbox{ for }k=1,\ldots,n.\]
The set of the behaviour descriptions ${\bf W}'$ for the $\PDC$ model we are constructing is ${\bf N}_n$. 

Let $w\in{\bf N}_n$ and $\tau\in(\min\sigma_0,\infty)$. Let $Q$ be a state variable occurring in $\varphi$. Then 
\[\ell=0\vee(\pred{Q};\top)\vee(\pred{\neg Q};\top),\ell=0\vee\ell=\infty\vee(\top;\pred{Q})\vee(\top;\pred{\neg Q})\in{\bf H},\]
because these formulas are instances of $T1$ and $T2$, respectively. This entails that 
\[M,w,[\tau,\tau+1]\models(\ell_{[Q]}=\ell\wedge\ell\not=0;\top)\vee(\ell_{[\neg Q]}=\ell\wedge\ell\not=0;\top)\]
and
\[M,w,[\min\sigma_0,\tau]\models(\top;\ell_{[Q]}=\ell\wedge\ell\not=0)\vee(\top;\ell_{[\neg Q]}=\ell\wedge\ell\not=0),\]
which implies that there are some $\xi,\eta\in{\bf R}$ such that $\xi<\tau<\eta$ and 
\[M,w,[\tau,\eta]\models\ell_{[Q]}=\ell\vee\ell_{[\neg Q]}=\ell\mbox{ and }M,I,[\xi,\tau]\models\ell_{[Q]}=\ell\vee\ell_{[\neg Q]}=\ell.\]
Let us fix some $\xi$ and $\eta$ with this property and denote the open neighbourhood $(\xi,\eta)$ of $\tau$ by $O_{Q,w,\tau}$. Similarly,

\[M,w,[\min\sigma_0,\min\sigma_0+1]\models(\ell_{[Q]}=\ell\wedge\ell\not=0;\top)\vee(\ell_{[\neg Q]}=\ell\wedge\ell\not=0;\top)\]
and hence there is an $\eta>\min\sigma_0$ such that 
\[M,w,[\min\sigma_0,\eta]\models\ell_{[Q]}=\ell\vee\ell_{[\neg Q]}=\ell.\]
We fix such an $\eta$ and write $O_{Q,w,\min\sigma_0}$ for the
semi-open neighbourhood $[\min\sigma_0,\eta)$ of
$\min\sigma_0$. Obviously
\[\bigcup\limits_{\tau\in[\min\sigma_0,\infty)}O_{Q,w,\tau}=[\min\sigma_0,\infty).\]
Moreover, ${\bf
O}_{Q,w}=\{O_{Q,w,\tau}:\tau\in[\min\sigma_0,\infty)\}$ is a
(relatively) open covering of $[\min\sigma_0,\infty)$. Here follows
the key observation in this proof: the compactness of the intervals of
the form $[\min\sigma_0+k,\min\sigma_0+k+1]$ where $k=0,1,2,\ldots$
implies that for every such $k$ there is a finite sub-covering ${\bf
O}_{Q,w,k}\subset{\bf O}_{Q,w}$ of
$[\min\sigma_0+k,\min\sigma_0+k+1]$. Let ${\bf
O}_{Q,w,k}=\{O_{Q,w,\tau_{Q,w,k,1}},\ldots,O_{Q,w,\tau_{Q,w,k,n_{w,k}}}\}$. We
will use the time points $\tau_{Q,w,k,i}$,
$i=1,\ldots,n_{w,k}$, $k=0,1,\ldots$, where $Q$ is a state variable
occurring in $\varphi$ to define an interpretation $(I')^w$ of ${\bf
L}$ in our $\PDC$ model under construction which corresponds to $I^w$
for $w\in{\bf W}'$. Let us denote the set of these time points by
$C_{Q,w}$. Since $\min\sigma_0\in C_{Q,w}$ and $C_{Q,w}\cap\sigma$ is
finite for every bounded interval $\sigma$, the set $C_{Q,w}\cap
[\min\sigma_0,\tau]$ contains a greatest time point for every
$\tau\in[\min\sigma_0,\infty)$. $(I')^w$ is defined by the following
clauses
\[\begin{tabular}{lp{4.5in}}
$(I')^w(s)=I(s)$ & for all symbols $s\in{\bf L}$ which are not state variables;\\
$(I')^w(Q)(\tau)=0$ & for all state variables $Q\in{\bf L}$ which do not occur in $\varphi$ and all $\tau\in\overline{\bf R}$;\\
$(I')^w(Q)(\tau)=1$ & for state variables $P$ which occur in $\varphi$ and $\tau$ such that $M,w,[\tau',\sup O_{Q,w,\tau'}]\models\ell_{[Q]}=\ell$, where $\tau'=\max(C_{Q,w}\cap [\min\sigma_0,\tau])$;\\
$(I')^w(Q)(\tau)=0$ & for state variables $Q$ which occur in $\varphi$ and $\tau$ such that $M,w,[\tau',\sup O_{Q,w,\tau'}]\models\ell_{[\neg Q]}=\ell$, where $\tau'$ is as above and for $\tau<\min\sigma_0$ as well.\\
\end{tabular}
\]
A straightforward argument based on the presence of the appropriate
instances of $\DC 1$-$\DC 6$ in ${\bf H}$ implies that this definition
of $(I')^w$ is correct and $I'$ satisfies the equality
\[\textstyle (I')^w_\sigma(\int S)=I^w_\sigma(\ell_{[S]})\]
for all state expressions $S\in{\bf S}$ and all intervals
$\sigma\in\tilde{\bf I}(\overline{R})$ such that
$\min\sigma_0\leq\min\sigma$.

The functions $(P')^w$, $w\in{\bf W}'$, are defined using the
respective $P^w$ by the equality
\begin{equation}\label{pprimedef}
(P')^w(\tau, A\cap{\bf W}')=P^w(\tau,A)
\end{equation}
for $w\in\bigcup_{i=0}^{n-1}{\bf N}_i$ and
$\tau\geq\min\sigma$.  Since $M,w_0,\sigma_0\models\psi_n$, the
construction of ${\bf W}'$ implies that $P^w(\tau,({\bf
W}')_{w,\tau})=1$ for all such $w$. Hence if
$P(\tau,A_1)\not=P(\tau,A_2)$, then $P(\tau,A_1\cap{\bf
W}'_{w,\tau})\not=P(\tau,A_2\cap{\bf W}'_{w,\tau})$ as well, which
implies that $A_1\cap({\bf W}')_{w,\tau}\not=A_2\cap({\bf
W}')_{w,\tau}$. That is why the equality (\ref{pprimedef}) defines the
function $(P')^w$ correctly.  We allow $(P')^w$ to be arbitrary for
$w\in{\bf W}'\setminus\bigcup_{i=0}^{n-1}{\bf N}_i$ , because
the truth values of formulas of probability height up to $n$ at
$w_0,\sigma_0$ do not depend on these functions.
  
Let $M'=\langle F_{\bf R},{\bf W}',I',P'\rangle$. An induction on $k$
implies that if $\psi$ is a $\PDC$ formula written in ${\bf L}$,
$h(\psi)\leq k$, $w\in{\bf N}_i$, $\sigma\in\tilde{\bf
I}(\overline{\bf R})$, $\min\sigma\geq\min\sigma_0$ and $k+i\leq n$,
then
\[M',w,\sigma\models\psi\mbox{ iff } M,w,\sigma\models\psi'\mbox{ and }P^w(\max\sigma,\sem{\psi'}_{M,w,\sigma})=(P')^w(\tau,\sem{\psi}_{M',w,\sigma}).\]
This, in particular, implies that
\[M',w_0,\sigma_0\models\varphi\mbox{ or }M',w_0,\sigma_0\models(\varphi;\ell=\infty).\]
In the latter case $M',w_0,\sigma\models\varphi$ for some $\sigma\in{\bf I}^\fin(\overline{\bf R})$ such that $\min\sigma=\min\sigma_0$.

This concludes the proof of the relative completeness of the axioms $\DC 1$--$\DC 6$, $T1$ and $T2$ for $\PDC$, because we have shown that the assumption that a given $\PDC$ formula is consistent with this proof system entails that the formula is satisfiable at a $\PDC$ model.

\section{$\PITL$ with infinite intervals and $\PNL$}
\label{itlvsnl}

The system which is closest to $\PITL$ both in its semantics and proof system is the probabilistic extension of neighbourhood logic $\PNL$ which was proposed in \cite{Gue00probab}. The modalities $\Diamond_l$ and $\Diamond_r$ of $\NL$ are defined by the clauses:
\[\begin{tabular}{ll}
$M,\sigma\models\Diamond_l\varphi$ & iff $M,\sigma'\models\varphi$ for some $\sigma'$ such that $\max\sigma'=\min\sigma$\\
$M,\sigma\models\Diamond_r\varphi$ & iff $M,\sigma'\models\varphi$ for some $\sigma'$ such that $\min\sigma'=\max\sigma$\\
\end{tabular}
\]
$\Diamond_l$ and $\Diamond_r$ are called {\em expanding} modalities
because they allow access outside the reference interval. The dual
modalities $\Box_d$ of $\Diamond_d$ are defined by the clauses
\[\Diamond_d\rightleftharpoons\neg\Diamond_r\neg\Diamond_d\varphi\]
for $d\in\{l,r\}$.

A duration calculus on the basis of $\NL$ was developed in
\cite{RZ97}. Infinite intervals are an alternative way to achieve the
expressivity of $\Diamond_r$. A truth preserving translation from
$\ITL$ with infinite intervals to $\NL$ is impossible for the trivial
reason that $\NL$ does not have infinite intervals and there is no
straightforward way to capture the $\ITL$ interpretation of flexible
symbols at infinite intervals. Furthermore, $\NL$ duration domains
known from the literature do not include $\infty$, but include
negative durations. However, if the only flexible symbols in the
considered vocabularies are $\ell$ and state variables, then the
duration calculi based on $\NL$ and on $\ITL$ with infinite intervals,
respectively, can be related by means of a translation which has the
following property:

\begin{quote}
If $\psi$ is the $\NL$-based $\DC$ formula which is the translation of some $\ITL$-based $\DC$ formula $\varphi$ and $FV(\varphi)=\{x_1,\ldots,x_n\}$, then
\begin{equation}\label{itltonl}
M',[\tau,\tau]\models\psi\mbox{ iff }M,[\tau,\infty]\models\varphi,\end{equation}
\end{quote}
where the duration domain of the $\ITL$ model $M$ is obtained from
that of the $\NL$ model $M'$ by removing the negative elements and
adding $\infty$, and the meanings of the non-logical symbols in $M$
and $M'$ on the intersection of the two duration domains are the
same. We describe such a translation in this section.

The predicate logic equivalences 
\[R(t_1,\ldots,t_n)\Leftrightarrow\exists x_1\ldots\exists x_n\left(R(x_1,\ldots,x_n)\wedge\bigwedge\limits_{i=1}^n t_i=x_i\right)\]
and
\[f(t_1,\ldots,t_n)=z\Leftrightarrow\exists x_1\ldots\exists x_n \left(f(x_1,\ldots,x_n)=z\wedge\bigwedge\limits_{i=1}^n t_i=x_i\right),\]
where $x_1,\ldots,x_n$ do not occur in $t_1,\ldots,t_n$, allow us to
assume that all atomic subformulas of the $\ITL$ formulas to be
translated are either rigid of have the form $\int S=x$ where $x$ is a
variable. We can also treat $\ell$ as $\int {\bf 1}$. The clauses
below define two auxiliary translations $(.)^\fin$ and $(.)^\infin$
from $\ITL$-based to $\NL$-based $\DC$.  $(.)^\fin$ translates an
$\ITL$ formula which is to be evaluated at a {\em finite} interval
into its $\NL$ equivalent. $(.)^\infin$ translates an $\ITL$ formula
which is to be evaluated at an {\em infinite} interval $\sigma$ into a
corresponding $\NL$ formula which defines the same condition on
$\sigma$ when evaluated at the zero-length interval
$[\min\sigma,\min\sigma]$. $(.)^\infin$ refers to $(.)^\fin$ for the
translation of $(.;.)$-formulas. Both auxiliary translations are
correct only under the assumption that the free variables of the given
$\ITL$ formulas range over non-negative finite durations. Infinity is
handled only where explicitly denoted by the symbol $\infty$. Atomic
formulas $R(t_1,\ldots,t_n)$ with the parameter list $t_1,\ldots,t_n$
consisting of individual variables and, possibly, $\infty$ translate
into dedicated {\em specialising} formulas $S^R_{t_1,\ldots,t_n}$,
which define the appropriate predicates on the non-$\infty$ parameters
according to the intended meaning of $R$ and the positions of the
occurrences of $\infty$ in $t_1,\ldots,t_n$. For instance, $S^=_{x,y}$
is $x=y$, $S^=_{x,\infty}$ is $\bot$, and $S^=_{\infty,\infty}$ is
$\top$. Atomic formulas with $=$ and function symbols are handled
similarly, e.g. the formula $S^+_{x,\infty;y}$ for $x+\infty=y$ is
$\bot$, and $S^+_{x,\infty;\infty}$ is $\top$.

\[\begin{tabular}{lll}
$\bot^\fin$ & $\rightleftharpoons$ & $\bot$\\
$(R(t_1,\ldots,t_n))^\fin$ & $\rightleftharpoons$ & $S^R_{t_1,\ldots,t_n}$\\
$(f(t_1,\ldots,t_n)=t_{n+1})^\fin$ & $\rightleftharpoons$ & $S^f_{t_1,\ldots,t_n;t_{n+1}}$\\
$(\int S=\infty)^\fin$ & $\rightleftharpoons$ & $\bot$\\
$(\int S=x)^\fin$ & $\rightleftharpoons$ & $\int S=x$\\
$(\varphi\Rightarrow\psi)^\fin$ & $\rightleftharpoons$ & $\varphi^\fin\Rightarrow\psi^\fin$\\
$(\varphi;\psi)^\fin$ & $\rightleftharpoons$ & $\exists x\exists y(\int {\bf 1}=x+y\wedge\Diamond_l\Diamond_r(\ell=x\wedge\varphi^\fin\wedge\Diamond_r(\ell= y\wedge\psi^\fin)))$\\
$(\exists x\varphi)^\fin$ & $\rightleftharpoons$ & $([\infty/x]\varphi)^\fin\vee\exists x(x\geq 0\wedge\varphi^\fin)$\\\\
$\bot^\infin$ & $\rightleftharpoons$ & $\bot$\\
$(R(t_1,\ldots,t_n))^\infin$ & $\rightleftharpoons$ & $S^R_{t_1,\ldots,t_n}$\\
$(f(t_1,\ldots,t_n)=t_{n+1})^\infin$ & $\rightleftharpoons$ & $S^f_{t_1,\ldots,t_n;t_{n+1}}$\\
$(\int S=\infty)^\infin$ & $\rightleftharpoons$ & $\forall x\Diamond_r\int S>x$\\
$(\int S=x)^\infin$ & $\rightleftharpoons$ & $\Diamond_r(\int S=x\wedge\Box_r\int S=0)$\\
$(\varphi\Rightarrow\psi)^\infin$& $\rightleftharpoons$ & $\varphi^\infin\Rightarrow\psi^\infin$\\
$(\varphi;\psi)^\infin$ & $\rightleftharpoons$ & $\Diamond_r(\varphi^\fin\wedge\Diamond_r(\ell=0\wedge\psi^\infin))$\\
$(\exists x\varphi)^\infin$ & $\rightleftharpoons$ & $([\infty/x]\varphi)^\infin\vee\exists x(x\geq 0\wedge \varphi^\infin)$\\
\end{tabular}
\]
As mentioned above, $(.)^\infin$ is correct only under the assumption that the free variables of the given $\ITL$ formulas range over non-negative finite durations. 
To remove this restriction, given an $\ITL$ formula $\varphi$ whose free variables are $x_1,\ldots,x_n$, we define the sequence of formulas $\varphi_0,\ldots,\varphi_n$ by the clauses 
\[\varphi_0\rightleftharpoons\varphi\mbox{ and }\varphi_i\rightleftharpoons(x_i\geq 0\wedge\varphi_{i-1})\vee[\infty/x_i]\varphi_{i-1}\mbox{ for }i=1,\ldots,n,\]
and choose the formula $\psi$ from (\ref{itltonl}) to be $(\varphi_n)^\infin$.
This translation can be extended to one between $\PDC$ with infinite intervals and a system of probabilistic $\DC$ based on $\NL$ by putting 
\[\begin{tabular}{lll}
$(p(\varphi)=x)^\fin$ & $\rightleftharpoons$ & $p(\varphi^\infin)=x$.\\
$(p(\varphi)=x)^\infin$ & $\rightleftharpoons$ & $\varphi^\infin\wedge x=1\vee\neg\varphi^\infin\wedge x=0$.
\end{tabular}
\]
A translation from $\NL$ into $\ITL$ with infinite intervals is possible too under the assumption that there is a time point $\tau_0$ such that the values of all flexible symbols except $\ell$ at intervals starting before $\tau_0$ are irrelevant to the truth value of the translated formula. This restriction is necessary, because an $\ITL$ formula cannot express conditions on the past prior to the beginning of the infinite reference interval. It can be avoided if one considers a system of $\ITL$ with intervals which can be infinite into the past as well, which is beyond the scope of this paper. If a property does not depend on the interpretation of the flexible symbols on the left of the beginning of the reference interval and can be expressed by an $\NL$ formula, then it can be expressed by an $\NL$ formula in which the only occurrences of $\Diamond_l$ are in subformulas of the form $\Diamond_l\Diamond_r\chi$. Given an $\NL$ formula $\varphi$ which satisfies this syntactical restriction, one can find an $\ITL$ formula $\psi$ such that $M,[\tau_0,\infty]\models\psi$ is equivalent to the existence of a $\tau_1\geq\tau_0$ such that $M',[\tau_0,\tau_1]\models\varphi$. 
Below we give a translation which, given a $\varphi$ of the form
\[\varphi::=\bot\mid R(t,\ldots,t)\mid (\varphi\Rightarrow\varphi)\mid\Diamond_r\varphi\mid\Diamond_l\Diamond_r\varphi\mid\exists x(x\geq 0\wedge\varphi)\]
produces a corresponding $\psi$. This translation produces formulas constructed using $\exists$, $\Rightarrow$, $\bot$, rigid formulas and formulas of the form
\begin{equation}\label{nltoitl}
(\ell=t_1;\ell=t_2\wedge\alpha;\top)
\end{equation}
with $\alpha$ being a modality-free formula.
The translation works by reducing the number of the occurrences of $\Diamond_l\Diamond_r$ and $\Diamond_r$ in formulas of the form (\ref{nltoitl}), yet with $\alpha$ being a $\NL$ formula. The $\ITL$ formula $\psi$ is obtained by starting from $(\ell=0;\ell=0\wedge\Diamond\varphi;\top)$. To understand the correctness of the translation, one can think of a system which has all the modalities $(.;.)$, $\Diamond_l$ and $\Diamond_r$, with the obvious semantics, and check that the translation rules correspond to valid equivalences at infinite reference intervals, provided that the free variables of the involved formulas have finite non-negative values. Here follow the transformation rules which define the translation:

\qquad

\noindent

\begin{tabular}{l}
$(\ell=t_1;\ell=t_2\wedge(\chi_1\Rightarrow\chi_2);\top)\rightarrow(\ell=t_1;\ell=t_2\wedge\chi_1;\top)\Rightarrow(\ell=t_1;\ell=t_2\wedge\chi_2;\top)$\\
$(\ell=t_1;\ell=t_2\wedge\Diamond_r\chi;\top)\rightarrow\exists z(\ell=t_1+t_2;\ell=z\wedge\chi;\top)$\\
$(\ell=t_1;\ell=t_2\wedge\Diamond_l\Diamond_r\chi;\top)\rightarrow\exists z(\ell=t_1;\ell=z\wedge\chi;\top)$\\
$(\ell=t_1;\ell=t_2\wedge\exists x(x\geq 0\wedge \chi);\top)\rightarrow\exists x(x<\infty\wedge (\ell=t_1;\ell=t_2\wedge\chi;\top))$ 
\end{tabular}

\qquad

The individual variable $z$ in the rules above is supposed to be fresh. The last rule can be applied only if $x\not\in FV(t_1),FV(t_2)$. This translation can be extended to one from $\PNL$ to $\PITL$ by mapping $\NL$ probability terms $p(\varphi)$ to $\PITL$ corresponding probability terms $p(\psi)$ where $\psi$ is the translation of $\varphi$.

\section*{Concluding remarks}

We conclude by discussing some restrictions on the scope of the completeness results about $\PITL$ and $\PDC$ presented in this paper. 

\subsubsection*{Countable additivity of probability functions}

According to our definition, the probability functions in $\PITL$ models are required to be just finitely additive, whereas classical probability theory is about countably additive probability functions. One simple reason for this is the choice to have an abstract domain of probabilities which is not required to be Dedekind-complete and therefore the infinite sums which are relevant to countable additivity cannot be guaranteed to exist. The difficulty in axiomatising countable additivity becomes even more obvious from the observation that $\PITL$ has the {\em L\"owenheim-Skolem} property. This means that countably-infinite consistent sets of $\PITL$ formulas can be satisfied at countably-infinite models, which, in particular, have countably-infinite domains. This follows immediately from the construction of the $\PITL$ model in the completeness argument for our proof system. Countably-infinite $\PITL$ models with countably additive probability functions validate formulas of the form
\[\forall x (p(\varphi)=0)\Rightarrow p(\exists x\varphi)=0.\]
This follows immediately from the fact that $x$ ranges over a countably-infinite domain. Hence, the above formula should be a theorem in a proof system which is complete with respect to models with countably additive probability functions, as long as the L\"owenheim-Skolem property holds. However, this formula is not valid in arbitrary models.

\subsubsection*{Completeness of $\PDC$ relative to (non-probabilistic) real-time $\ITL$}

Our demonstration that some well-known axioms of (non-probabilistic) $\DC$ form a proof system which is complete relative to {\em probabilistic} $\ITL$ with infinite intervals was hardly a technical challenge, given the similar proofs from \cite{HZ92,RZ97}. It would have been interesting to develop a proof system for $\PDC$ which is complete relative to real-time $\ITL$ without probabilities. The proof of Lemma \ref{greatlemma}, which is the key step in our model construction for the completeness argument for $\PITL$, explains why this is impossible. The model construction involves an expression of $\tau$-equivalence by the formulas \begin{equation}\label{tauequiv}
(\Box\forall(\chi^\nu\Leftrightarrow\chi^{\nu'})\wedge\ell=c;\ell=\infty)
\end{equation}
for $\tau$ being the equivalence class $[c]$ of the rigid constant $c$. The relation of $\tau$-equivalence is needed to hold between any given $w\in{\bf W}$ from a $\PDC$ model $M=\langle F_{\bf R},{\bf W},I,P\rangle$ and the $v\in{\bf W}$ which are needed to populate $\sem{\varphi}_{M,w,\sigma}$ for $\varphi$ such that $M,w$ is supposed to satisfy $p(\varphi)\not=0$ at intervals $\sigma$ whose end point is $\tau$. The proof of Lemma \ref{greatlemma} relies on the possibility to use the formulas (\ref{tauequiv}) and an assumption which essentially amounts to the derivability of $\neg\varphi$ from some appropriately chosen formulas in order to derive the existence of a formula $\theta$ such that the same formulas imply $(\theta\wedge\ell=c;\ell=\infty)\Rightarrow\neg\varphi$, which in its turn enables an application of the $\PITL$ proof rule $P_\leq$ to derive $\theta\Rightarrow p(\varphi)=0$ and reach the aimed contradiction. The existence of the formula $\theta$ amounts to the interval-related intepolation property of $\ITL$ with infinite intervals (see Section \ref{interpolation}). Unfortunately, $\DC$ has neither this interpolation property, nor the related Craig interpolation property \cite{Gue04b}. The counterexample to Craig interpolation in \cite{Gue04b} indicates that the property could possibly be restored by allowing infinitary formulas to take the role of $\theta$. $\DC$ is not a compact logic and therefore derivability from infinite sets of premises is not reducible to derivability from finite ones. Hence, in order to achieve sufficient deductive power, the proof rule $P_\leq$ would have to be replaced by one allowing infinitary formulas on the left of $\Rightarrow$ as well.  The deductive power of a finitary rule would be insufficient for the role of $P_\leq$ in any presumable finitary proof system for $\PDC$ that is complete relative to (non-probabilistic) real-time $\ITL$ with infinite intervals.

\bibliographystyle{alpha}
\bibliography{../bibfiles/mybiblio}

\end{document}